\documentclass[letterpaper,11pt]{article}
\usepackage[toc,page]{appendix}
\usepackage[margin=1in]{geometry}
\usepackage[bookmarks, colorlinks=true, plainpages = false, citecolor = blue,linkcolor=red,urlcolor = blue, filecolor = blue,pagebackref]{hyperref}
\usepackage{url}
\usepackage{amsmath,amsfonts,amsthm,amssymb,bm,verbatim,dsfont,mathtools}
\usepackage{color,graphicx,appendix}
\usepackage{subfigure}
\usepackage{etoolbox}
\usepackage{tikz}
\usepackage{xr,xspace}
\usepackage{todonotes}
\usepackage{paralist}
\usepackage{caption,soul}
 \usepackage{algorithm}
 \usepackage{algorithmic}
\usepackage{appendix}
\makeatletter

\newtheorem{theorem}{Theorem}
\newtheorem*{theorem*}{Theorem}
\newtheorem{lemma}{Lemma}
\newtheorem{proposition}{Proposition}
\newtheorem{corollary}{Corollary}
\theoremstyle{definition}
\newtheorem{definition}{Definition}

\newtheorem{claim}{Claim}

\newtheorem{remark}{Remark}
\newtheorem{assumption}{Assumption}

\usepackage{xspace,prettyref}

\newcommand{\iiddistr}{{\stackrel{\text{\iid}}{\sim}}}

\newcommand{\reals}{{\mathbb{R}}}
\newcommand{\integers}{{\mathbb{Z}}}
\newcommand{\naturals}{{\mathbb{N}}}

\newcommand{\identity}{\mathbf I}

\newcommand{\diff}{{\rm d}}

\newcommand{\expect}[1]{\mathbb{E}\left[ #1 \right]}

\newcommand{\prob}[1]{ \mathbb{P}\left\{ #1 \right\} }

\newcommand{\var}{\mathsf{var}}

\newcommand{\Bern}{{\rm Bern}}

\newcommand{\eg}{e.g.\xspace}
\newcommand{\ie}{i.e.\xspace}
\newcommand{\iid}{i.i.d.\xspace}
\newrefformat{eq}{(\ref{#1})}
\newrefformat{chap}{Chapter~\ref{#1}}
\newrefformat{sec}{Section~\ref{#1}}
\newrefformat{alg}{Algorithm~\ref{#1}}
\newrefformat{fig}{Fig.~\ref{#1}}
\newrefformat{tab}{Table~\ref{#1}}
\newrefformat{rmk}{Remark~\ref{#1}}
\newrefformat{clm}{Claim~\ref{#1}}
\newrefformat{def}{Definition~\ref{#1}}
\newrefformat{cor}{Corollary~\ref{#1}}
\newrefformat{lmm}{Lemma~\ref{#1}}
\newrefformat{prop}{Proposition~\ref{#1}}
\newrefformat{app}{Appendix~\ref{#1}}
\newrefformat{hyp}{Hypothesis~\ref{#1}}
\newrefformat{thm}{Theorem~\ref{#1}}
\newrefformat{ass}{Assumption~\ref{#1}}

\newcommand{\pth}[1]{\left( #1 \right)}
\newcommand{\qth}[1]{\left[ #1 \right]}
\newcommand{\sth}[1]{\left\{ #1 \right\}}

\newcommand{\abth}[1]{\left | #1 \right |}

\newcommand{\norm}[1]{\left\|{#1} \right\|}

\newcommand{\lnorm}[2]{\left\|{#1} \right\|_{{#2}}}

\newcommand{\Fnorm}[1]{\lnorm{#1}{\rm F}}
\newcommand{\fnorm}[1]{\|#1\|_{\rm F}}
\newcommand{\opnorm}[1]{\left\| #1 \right\|}
\newcommand{\iprod}[2]{\left \langle #1, #2 \right\rangle}
\newcommand{\Iprod}[2]{\langle #1, #2 \rangle}
\newcommand{\indc}[1]{{\mathbf{1}_{\left\{{#1}\right\}}}}

\newcommand{\diag}[1]{\mathsf{diag} \left\{ {#1} \right\} }

\newcommand{\calA}{{\mathcal{A}}}

\newcommand{\calE}{{\mathcal{E}}}

\newcommand{\calN}{{\mathcal{N}}}

\newcommand{\calR}{{\mathcal{R}}}
\newcommand{\calS}{{\mathcal{S}}}
\newcommand{\calT}{{\mathcal{T}}}
\newcommand{\calU}{{\mathcal{U}}}

\newcommand{\calW}{{\mathcal{W}}}
\newcommand{\calX}{{\mathcal{X}}}

\newcommand{\Tr}{\mathsf{Tr}}

\renewcommand{\hat}{\widehat}
\renewcommand{\tilde}{\widetilde}

\begin{document}

\title{Securing Distributed Gradient Descent in High Dimensional Statistical Learning} 

\date{\today}

\author{
Lili Su \\
CSAIL, EECS\\
MIT\\
{lilisu@mit.edu}
\and 
Jiaming Xu\\
The Fuqua School of Business\\
Duke University\\
{jiaming.xu868@duke.edu}
}

\maketitle

\begin{abstract}
We consider unreliable distributed learning systems wherein the training data is kept confidential by external workers, and the learner has to interact closely with those workers to train a model. In particular, we assume that there exists a system adversary that can adaptively compromise some workers; the compromised workers deviate from their local designed specifications by sending out arbitrarily malicious messages. 

We assume in each communication round, up to $q$ out of the $m$ workers suffer Byzantine faults. Each worker keeps a local sample of size $n$ and the total sample size is $N=nm$. We propose a secured variant of the gradient descent method that can tolerate up to a constant fraction of Byzantine workers, i.e., $q/m = O(1)$. Moreover, we show the statistical estimation error of the iterates converges in $O(\log N)$ rounds to $O(\sqrt{q/N} + \sqrt{d/N})$, where $d$ is the model dimension. As long as $q=O(d)$, our proposed algorithm achieves the optimal error rate $O(\sqrt{d/N})$. Our results are obtained under some technical assumptions. Specifically, we assume strongly-convex population risk. Nevertheless, the empirical risk (sample version) is allowed to be non-convex. The core of our method is to robustly aggregate the gradients computed by the workers based on the filtering procedure proposed by Steinhardt et al. \cite{Steinhardt18}. On the technical front, deviating from the existing literature on robustly estimating a finite-dimensional mean vector, we establish a {\em uniform} concentration of the sample covariance matrix of  gradients, and show that the aggregated gradient, as a function of model parameter, converges uniformly to the true gradient function. To get a near-optimal uniform concentration bound, we develop a new matrix concentration inequality, which might be of independent interest. 
\end{abstract}

\section{Introduction}
\label{sec: intro}
Distributed machine learning has been an attractive solution to large-scale problems for years \cite{boyd2011distributed}. At the same time, learning in the presence of (possibly malicious) outliers has a deep root in robust statistics~\cite{huber2011robust} and has become an extremely active area recently~\cite{7782980,lai2016agnostic,Charikar:2017:LUD:3055399.3055491,DiakonikolasKK017,Steinhardt18}. However, most of the previous work implicitly assumes that the systems used to carry out the learning task are reliable, i.e., each computing device follows some designed specification. In this work, we consider unreliable distributed learning systems \cite{Lynch:1996:DA:2821576,su2017defending,Chen:2017:DSM:3175501.3154503,Blanchard2017,alistarh2018byzantine,yin2018byzantine} that are prone to system failures or even adversarial attacks. 
In particular, we assume that there exists a system adversary that can adaptively choose some computing devices to compromise; the compromised devices deviate from their local designed specifications and behave maliciously in an arbitrary manner. 

Our consideration of unreliable distributed learning systems is motivated by the recent trends in a new learning framework wherein the training data is kept confidential by external computing devices, and the learner interacts with the external computing devices  to train a model.  
In classical learning frameworks, 
data is collected from its providers (who may or may not be voluntary) and is stored by the learner. 
Such data collection immediately leads to  data providers' serious privacy concerns, which  
root in not only purely psychological reasons but also the poor real-world practice of privacy-preserving solutions. 
In fact, privacy breaches occur frequently, with recent examples including Facebook data leak scandal, iCloud leaks of celebrity photos, and PRISM surveillance program. 
Putting this privacy risk aside, data providers often benefit from the learning outputs. For example, in medical applications, although participants may be embarrassed about their use of drugs, they might benefit from good learning outputs that can provide high-accuracy predictions of developing diseases. 

To resolve this dilemma of data providers, researchers and practitioners have proposed an alternative learning framework wherein the training data is kept confidential by its providers from the learner and these providers function as workers 
\cite{konevcny2015federated,federatedlearningblog,Duchi:2014:PAL:2700084.2666468}. This framework has been implemented in 
 practical systems such as Google's {\em Federated Learning} \cite{konevcny2015federated,federatedlearningblog}, wherein Google tries to learn a model with the training data kept confidential on the users' mobile devices.  We refer to this new learning framework as {\em learning with external workers}.
In contrast to the traditional learning framework under which models are trained within data-centers, in {\em learning with external workers} the learner faces  
 serious {\em security} risk: (1) some external workers may be highly unreliable or even be malicious (hacked by the system adversary);  
(2) the learner lacks enough administrative power over those external workers. 
%
In this paper, we aim to develop strategies to safeguard distributed machine learning against adversarial workers while keeping the following two key practical constraints in mind: \footnote{Depending on applications, there might be many other constraints such as unevenly distributed training data, intermittent availability of external workers, etc. In addition, different applications might even call for different performance metrics. We would like to explore these more richer settings in our future work. 
}
\begin{itemize}
\item Small local samples versus high model dimensions:
While the total volume of data over all workers may be large, individual workers may keep only small samples comparing to model dimensions. That is, the training data is {\em locally} a scarce resource. 

\item Communication constraints:  Similar to other large distributed systems, the external workers are typically highly heterogeneous in terms of computation powers, real-time local computation environments, etc. As a result of this, each round of communication requires synchronization \cite{JMLR:v14:zhang13b};  the transmission between the external workers and the learner typically suffers from high latency and low throughout. 
\end{itemize}
These two constraints together raise significant challenges for designing securing strategies. Without the first constraint, a one-shot outlier-resilient aggregation procedure suffices: each worker separately performs learning based on the local sample and sends the local estimates to the learner who  aggregates these estimates to output a final global estimate. This procedure is straightforward to implement and is communication-efficient \cite{feng2014distributed,JMLR:v14:zhang13b}. 
However, the correctness of these algorithms crucially relies  on the assumption that the local sample size is sufficiently large. In particular, $n =N/m \gg d$, where $m$ is the number of workers, $n$ is the local sample size, $N=nm$ is the total sample size, and $d$ is the model dimension. In contrast, practical distributed learning systems often operate in the regime where $n\ll d$. Two immediate consequences are: (1) to learn an accurate model, the learner has to interact closely with those external workers, and such close interaction gives the adversary more chances to foil the learning process; (2) identifying the adversarial workers based on abnormality is highly challenging, as it becomes difficult to distinguish the statistical errors from the adversarial errors when the sample sizes are small. 
In addition, due to the randomness of the training data, the estimates computed at different rounds are highly dependent on each other. See Section \ref{subsubsec: uniform} for further explanation. 
 
There have been attempts to robustify stochastic gradient descent (SGD) 
\cite{Blanchard2017,alistarh2018byzantine}  
with different focus from what we consider here. In particular, \cite{Blanchard2017} assumes all the
workers can access the whole data sample. 
Similar to ours, the concurrent work \cite{alistarh2018byzantine} considers the scenario where data is generated and stored in a distributed fashion at the workers. However, \cite{alistarh2018byzantine} assumes that 
in each iteration the workers are able to use \emph{fresh data} to compute the gradients. However, fresh data in each round implies that the local sample size grows with time, which is not necessarily true in some applications. The fresh data assumption is crucial in their analysis: with fresh data, conditioning on the current model parameter estimator, the local gradients computed at different workers become independent, and the existing analysis of robust mean estimation may suffice. 
 In this work, we assume that the sample size is fixed over time\footnote{Fixed sample size arises, for example, when the model training speed is significantly faster than the data generation speed.}, and the training data is 
stored in a distributed fashion \cite{Chen:2017:DSM:3175501.3154503, yin2018byzantine}. 

\paragraph{Contributions:}
In this work, we propose a robust gradient descent method that tolerates up to a constant fraction of adversarial workers (i.e., $\frac{q}{m} = O(1)$) and converges to a statistical estimation error $O(\sqrt{q/N} + \sqrt{d/N})$ in $O(\log N)$ communication rounds; whereas, the minimax-optimal error rate in the failure-free and centralized setting is $O(\sqrt{d/N})$. \footnote{
See \cite[Section 3.2]{YW-ITSTATS} for a proof. 
Note that $O(\sqrt{d/N})$ is the minimax optimal estimation error rate 
without any additional structure on the model parameter.  When the model 
parameter has additional structure, such as sparsity, the $\sqrt{d}$ factor can possibly be improved.}
As long as $q=O(d)$, our proposed algorithm achieves the optimal error rate $O(\sqrt{d/N})$, matching the failure-free optimal error rate.
Our results are obtained under some technical assumptions that we hope to relax in the future. Specifically, we assume that the population risk is strongly-convex. Nevertheless, we do allow the empirical risk (sample version) to be non-convex. 

On the technical front, 
to deal with the interplay of the randomness of the data and the iterative updates of the model choice $\theta$, we first establish the concentration of sample covariance matrix of gradients {\em uniformly}\footnote{See \cite{raginskyece,shalev2014understanding} and reference therein for details about uniform convergence of functions. } at all possible model parameters; then we prove that our aggregated gradient, as a function of $\theta$, converges uniformly to the population gradient function $\nabla F(\cdot)$. 
Similar uniform concentration of sample covariance matrix has been derived in \cite[Lemma 2.1]{Charikar:2017:LUD:3055399.3055491} under the assumption that the gradients are sub-gaussian. While sub-gaussian \emph{data distribution} is commonly assumed in statistical learning literature, the resulting \emph{gradients} may be sub-exponential or even heavier tailed. For example, in the simplest linear regression example, the gradients are sub-exponential instead of sub-gaussian. 
Note that standard routine to bounding the 
spectral norm of the sample covariance matrix 
is available, 
 see \cite[Theorem 5.44]{vershynin2010introduction} and \cite[Corollary 3.8]{adamczak2010quantitative} for example. 
 However, it turns out that using these existing results, the uniform 
 concentration bound obtained is far from being optimal. To this end, 
 we develop a new concentration inequality for matrices with i.i.d.\ sub-exponential  column vectors. 
 This new inequality leads to a near-optimal uniform bound. 
 Relaxing the distributional assumption from sub-gaussian to sub-exponential is highly nontrivial. See \cite[Section 1.3]{vershynin2012close} and \cite{adamczak2010quantitative} for details. 
Our analysis framework developed in this work is not tied to sub-exponential assumptions. 
With different gradient distributional assumptions such as bounded second moment, one can follow our analysis roadmap to obtain
different uniform concentration bounds for the sample covariance matrix of gradient vectors, which in turn
implies different error bounds to our robust gradient descent method.

\vskip \baselineskip 

Note that in our algorithm, 
we let each non-faulty worker compute the local gradient based on the {\em entire} local sample (all $n$ data points). Since 
$n$ is small,  the computational burdens of the workers are reasonable. It has been demonstrated numerically in \cite{45648} that in the adversary-free setting, there is a performance improvement when each worker performs a few epochs of SGD before the model updates are aggregated. Whether there will be similar performance improvement in our adversary-prone setting is unclear, and we would like to leave this direction for future exploration.

\subsection{Comparison with Robustly Estimating a Finite-dimensional Mean Vector}
\label{subsec: comparison}
Our work is closely related to high-dimensional robust mean estimation -- a research area that has been
intensively studied 
\cite{7782980,lai2016agnostic,Charikar:2017:LUD:3055399.3055491,DiakonikolasKK017,Steinhardt18}. High-dimensional\footnote{The notion of ``high-dimension'' here does NOT refer to the setting where $d\gg N$. } robust mean estimation focuses on estimating the mean of a $d$-dimensional random vector 
from a contaminated dataset whose $\epsilon>0$ fraction of data is arbitrarily corrupted. 

Our algorithm uses robust mean estimation -- in particular, the procedure developed in~\cite{Steinhardt18} -- as a sub-routine to aggregate the gradients computed by the workers in each iteration.  
We provide a way to leverage robust mean estimation primitive to design an optimization algorithm that is resilient to adversarial system failures. 
Similar attempts were made in concurrent work \cite{diakonikolas2018sever,klivans2018efficient,prasad2018robust}.
In particular, \cite{klivans2018efficient} focuses 
on the linear or polynomial regression model, and \cite{diakonikolas2018sever,prasad2018robust} consider a general machine learning model similar to ours. 
The correctness of the robust gradient descent method proposed in \cite{prasad2018robust} relies on 
uniform approximation of the aggregated gradient to the true population gradient~\cite[Def. 1]{prasad2018robust}; however, 
only point-wise approximation bound 
is proved ~\cite[Lemma 1]{prasad2018robust}. 
In contrast, we prove a high-probability, uniform approximation bound
by assuming the local gradients are sub-exponential 
(See \prettyref{thm: uniform bound}). A similar uniform approximation bound
is proved in~\cite[Prop. B.5]{diakonikolas2018sever} with stronger assumptions. In fact, in proving ~\cite[Prop. B.5]{diakonikolas2018sever}, they essentially assume the local gradients are bounded\footnote{Recall that bounded random variables fall within the sub-gaussian family.} by $L'$ in $\ell_2$ norm and restrict
$N \ge q \gg d^2 (L')^2$,
where $L'$ is the Lipschitz continuous parameter of local gradients 
which may scale polynomially with $N,d$. See \prettyref{rmk:gradient_approx}
for detailed comparisons.

\subsubsection{Why uniform convergence?}
\label{subsubsec: uniform}
The existing analysis of robust mean estimation assumes that 
the good data vectors 
are independently and identically distributed,
and hence can only guarantee the performance of the gradient estimation at a {\em given} $\theta.$ 
However, in our problem, we need to take into account the fact that $\theta_t$ is updated iteratively based on the training data; due to the randomness of the training data, though $\theta_0$ might be treated as given, $\sth{\theta_t}_{t\ge 1}$ are random. As $\sth{\theta_t}_{t\ge 1}$ are obtained based on the same set of training data, 
they are highly dependent on each other. As a consequence, conditioning $\theta_t$ ($t \ge 1$), 
the gradients computed by the good workers are no longer $\iid$. This is in sharp contrast
to 
\cite{alistarh2018byzantine}. 

To deal with this interdependency and
unspecified behavior of adversarial workers, we instead view each
gradient as a \emph{function} of model parameter $\theta$, and aim to
robustly estimate the mean of the \emph{infinite-dimensional} 
gradient function -- the true population gradient function. 
This poses significant challenges in proving the desired robustness guarantees
and estimation error bounds. To this end, we  establish a 
uniform concentration of sample covariance matrix of gradient functions. 
To get a near-optimal uniform concentration bound, we develop a new matrix concentration inequality. 

\subsection{Further Related Work}
Recent years have witnessed a flurry of research on securing distributed machine learning algorithms against adversarial attacks.  Here we can only hope to cover a fraction of them
we see most relevant. See \cite{Chen:2017:DSM:3175501.3154503,Charikar:2017:LUD:3055399.3055491,feng2014distributed,Su:2016:FMO:2933057.2933105,DiakonikolasKK017,7782980,Blanchard2017,alistarh2018byzantine,yin2018byzantine} and references therein for more details. 
Both \cite{Su:2016:FMO:2933057.2933105,Blanchard2017} considered a pure optimization framework and characterizations of statistical performance of the learning outputs are left open; whereas \cite{Chen:2017:DSM:3175501.3154503} 
studied the same statistical learning framework as ours. In particular, \cite{Chen:2017:DSM:3175501.3154503} proposed an  algorithm that converges in logarithmic rounds to an estimation error  $O(\sqrt{dq/N})$ for $q \ge 1$,
which is suboptimal up to a multiplicative factor of $\sqrt{q}$. In the low dimensional regime where $d=O(1)$, the concurrent work \cite{yin2018byzantine} obtains an order-optimal error rate
based on coordinate-wise median and trimmed mean, but the dependency of error rate on dimension $d$ is 
highly suboptimal and is even inferior to the result in \cite{Chen:2017:DSM:3175501.3154503}. 
In this work, we focus on the more general regime of model dimension $d$. 

\paragraph{Notation} 
We use standard big $O$ notations, e.g., for any sequences $\{a_n\}$ and $\{b_n\}$, $a_n=O(b_n)$ or $a_n  \lesssim b_n$
if there is an absolute constant $C>0$ such that $a_n/ b_n \le C$.

\section{System Model}
\label{subsec: stat risk} 
Let $X$ denote the input data generated from some \emph{unknown} distribution $P$. 
Let $\Theta \subset \reals^d$ denote the set of all possible model parameters. We consider 
a risk function $f: \calX \times \Theta \to \reals$, where 
$f(x, \theta)$ measures the risk induced by a realization $x$ of the data under the model
parameter choice $\theta$. 
A classical example of the above statistical learning framework is linear regression, where $ x = (w, y) \in \reals^{d-1} \times \reals $ is the feature-response pair and $f(x, \theta) = \frac{1}{2}\left( \iprod{w}{\theta} - y \right)^2$.
The learner is interested in learning 
$\theta^*$ which minimizes the \emph{population risk} \ie, 
\begin{align}
\label{eq:min_pop_risk}
\theta^* \in \arg \min_{\theta \in \Theta} F(\theta) \triangleq \expect{f(X, \theta)}
\end{align}
-- assuming that  $F(\theta) \triangleq \expect{ f( X, \theta)}$  is well-defined for every $\theta \in \Theta$. The model choice $\theta^*$ is optimal in the sense that it minimizes the expected risk to pay if it is used for prediction. 
 
If the population risk $F$ 
were known, then 
$\theta^*$ might be computed by solving the minimization problem in \prettyref{eq:min_pop_risk}. 
In statistical learning, however, the distribution $P$ (thus the population risk $F$) is typically unknown; instead, training data is available for learning $\theta^*$.
Formally, we assume that there exist $N$ i.i.d.\ data points $X_i \iiddistr P$ in the decentralized learning system wherein the training data is kept locally by data providers  and cannot be accessed by the learner directly. The learner can only request those providers to compute gradient-like quantities of their locally kept data, as is the case in Federated Learning.  
We refer to those data providers as workers, as they can be viewed as ``recruited'' by the learner.
We assume there are $m$ workers, and the $N$ data points are distributed evenly across the $m$ workers. Specifically, the index set $[N]$ is partitioned into $m$ subsets $S_j$ such that  $|\calS_j|=N/m
\triangleq n$, and $\calS_i \cap \calS_j=\emptyset$ for $i\not=j$. \footnote{It would be interesting to consider more general data partitions, and we leave this as one future direction.}
Notably, 
$n$ is often much smaller than the model dimension $d$. 

The learner communicates with the workers in synchronous communication rounds, but non-faulty workers do not communicate with each other.
We leave the asynchronous communication as one future direction.  
We use the Byzantine fault model \cite{Lynch:1996:DA:2821576} to capture the unreliability and potential malicious behaviors of the workers. It is assumed that up to $q$ out of the $m$ workers suffer Byzantine faults and thus behave arbitrarily and possibly maliciously. Those faulty workers are referred to as Byzantine workers. 
The arbitrarily faulty behavior arises when the workers are reprogrammed 
by the system adversary. 
We assume the learner knows the upper bound $q$ -- a standard assumption in literature \cite{DiakonikolasKK017,Charikar:2017:LUD:3055399.3055491,Steinhardt18}. 
Nevertheless, an effective and efficient learning algorithm that does not call for the knowledge of $q$ as input is highly desirable. 
The set of Byzantine workers is allowed to {\em change} between communication rounds; the adversary can choose different sets of workers to control across communication rounds. 
Byzantine workers are assumed to have {\em complete knowledge} of the system, including the total number of workers $m$, all $N$ data points over the whole system, the programs that the workers are supposed to run, the program run by the learner, and the realization of the random bits generated by the learner. Moreover, Byzantine workers can collude. Nevertheless, when the adversary gives up the control of a worker, this worker recovers and becomes normal immediately. 
Note that this mobile Byzantine fault model is more general than the most classic Byzantine fault model, where the set of Byzantine workers is fixed throughout an execution. 


\section{Our Algorithm and Main Results}
\label{sec: alg_main}
A standard approach to estimate $\theta^*$ in statistical learning
is via empirical risk minimization. Given $N$ independent copies $X_1, \cdots, X_N$ of $X$,   
the empirical risk function is a random function over $\Theta$ defined as  
$(1/N) \sum_{i=1}^{N} f(X_i, \theta)$ for all $\theta \in \Theta.$
By the functional law of large numbers, the empirical risk function converges uniformly to 
the population risk function $F(\theta)$
in probability as sample size $N \to \infty$.\footnote{See \cite{raginskyece,shalev2014understanding} and reference therein for details about uniform convergence of functions.} 
As a consequence, we expect the minimizer of the empirical
risk function (which is random) also converges to the population risk minimizer $\theta^*$ in probability. 
While it may be possible to secure the empirical risk minimization using 
some ``robust" versions of empirical risk functions~\cite{Su:2016:FMO:2933057.2933105,Charikar:2017:LUD:3055399.3055491}, 
the characterizations of the estimation error are either unavailable or too loose. Moreover, in our distributed setting, it 
is costly to transmit the local empirical risk functions. Similar observation is made in \cite{prasad2018robust}. 
In this paper, we take a different approach: Instead of robustifying the empirical risk functions, we aim at robustifying the {\em learning process}. Specifically, we focus on securing the gradient descent method against the interruption caused by the Byzantine workers during model training. We focus on gradient descent as it is one of the most important and fundamental algorithms in machine learning
\cite{bertsekas2015convex,smola2008introduction}. At a high level, many machine learning problems are solved by minimizing certain 
appropriate risk (cost) function using gradient descent~\cite{mei2016landscape}. 
%

Recall that $\nabla F(\theta_{t-1})$ is the gradient of the population risk at $\theta_{t-1}$, $\eta$ is some fixed stepsize, and $\theta_0$ is the given initial guess of $\theta^*$. 
For the perfect gradient descent method, i.e., 
\begin{align}
\theta_t= \theta_{t-1} -\eta \times \nabla F(\theta_{t-1}),\label{eq:pop_gd}
\end{align} 
to converge exponentially fast, the following standard assumption is often adopted~\cite{boyd2004convex}. 
\begin{assumption}\label{ass:pop_risk_smooth}
The population risk function $F: \Theta \to \reals$ is $M$-strongly convex, and differentiable over $\Theta$ with $L$-Lipschitz gradient. That is, for all $\theta,\theta' \in \Theta$,
\begin{align*}
& F(\theta' )  \ge F(\theta) + \iprod{\nabla F(\theta)}{~\theta'-\theta} + \frac{M}{2} \norm{\theta'-\theta}^2, \\
& \norm{ \nabla F(\theta) - \nabla F(\theta')}  \le L \norm{\theta-\theta'}.
\end{align*}
\end{assumption}
Note that both $M$ and $L$ may scale in $d$ -- the dimension of $\theta$.
Though we assume strong-convexity of the population risk, the empirical risk can be highly non-convex. 
Detailed comments on strong convexity assumption can be found in Remark \ref{rmk:strongly-convex}.

Our approximate gradient descent method is given by Algorithm \ref{alg:BGD}.

\begin{algorithm}[htb]
\caption{Approximate Gradient Descent Method: Round $t\ge 1$}\label{alg:BGD}
\begin{center}
{\bf \color{blue}{\em The learner:}}
\end{center}
\begin{algorithmic}[1]
\STATE {\em Initialization:} Let $\theta_0$ be an arbitrary point in $\Theta$. Let $\eta = \frac{M}{2L^2}$.  
\STATE Broadcast the current model iterate $\theta_{t-1}$ to all workers;
\STATE Wait to receive all the gradients reported by the $m$ workers; 
Let $g_j(\theta_{t-1})$ denote the value received from worker $j$. \\
If no message from worker $j$ is received, set $g_j(\theta_{t-1})$ to be some arbitrary value;
\STATE  Aggregate gradients: Pass the received gradients to a gradient aggregator $\calR $ to obtain an aggregated gradient $G(\theta_{t-1})$, i.e., 
\begin{align}
G(\theta_{t-1})~ \gets ~\calR (g_1(\theta_{t-1}), \cdots, g_j(\theta_{t-1}), \cdots, g_m(\theta_{t-1})). \label{eq:robust_aggregation}
\end{align}
\vskip -0.2\baselineskip
\STATE Update: $ \theta_{t} ~ \gets ~  \theta_{t-1}   -\eta \times G(\theta_{t-1})$;
\end{algorithmic}
\begin{center}
{\bf \color{blue}{\em Worker $j$:}}
\end{center}
\begin{algorithmic}[1]
\STATE On receipt of $\theta_{t-1}$, compute the gradient at $\theta_{t-1}$, i.e., $\frac{1}{n}\sum_{i \in \calS_j}\nabla f (X_i, \theta_{t-1})$;
\STATE Send $\frac{1}{n}\sum_{i \in \calS_j}\nabla f (X_i, \theta_{t-1})$ back to the learner;
\end{algorithmic}
\end{algorithm}

If worker $j$ is non-faulty at round  $t$, on receipt of $\theta_{t-1}$, it 
computes its local 
gradient at $\theta_{t-1}$
and reports the computed gradient $g_j(\theta_{t-1})$ to the learner; 
if worker $j$ is Byzantine faulty at round $t$,
it reports an arbitrary value to the learner. Formally, 
 \begin{align*}
 g_j(\theta_{t-1}) = 
 \begin{cases}
 \frac{1}{n}\sum_{i \in \calS_j}\nabla f (X_i, \theta_{t-1}), &~~ \text{if $j$ is non-faulty  at round  $t$}; \\
 \star, &~~ \text{otherwise}, 
 \end{cases}
 \end{align*}
where $\star$ denotes an arbitrary value. 
Notably, the Byzantine workers might use all the information of the system to determine what value to report. The learner 
aggregates the received gradients via a gradient aggregator $\calR$ (an algorithmic function) to obtain an approximate gradient (line 4 of Algorithm \ref{alg:BGD}). 

%
%
We use a gradient aggregator that originates from robust mean estimation \cite{Steinhardt18}. 
%
%
%
%
For ease of exposition, we postpone the presentation of this gradient aggregator after stating our main results.

\subsection{Main Results}
To characterize the statistical estimation error rate of our proposed algorithm, we adopt some  
assumptions. Similar assumptions are made in~\cite{Chen:2017:DSM:3175501.3154503} and \cite{yin2018byzantine}. 
We illustrate our results by applying them to the classical linear regression and logistic regression problems in Section \ref{sec: linear regression}.

Though we successfully relax the distributional assumption from sub-gaussian 
to sub-exponential, no evidence so far hints that no further relaxation is possible. 
In fact, our analysis framework developed in this work is not tied to sub-exponential assumptions. 
With different gradient distributional assumptions such as bounded second moment, one can follow our analysis roadmap (Lemma 2) to obtain
different uniform concentration bounds for the sample covariance matrix of gradient vectors, which in turn
implies different error bounds to our robust gradient descent method. Detailed comments on our distributional assumptions can be found in Remark \ref{rmk:sub-exponential}. 

The concurrent work \cite{diakonikolas2018sever} attempts to prove uniform approximation bounds under bounded second moment assumption. However, their proof, in its current form, only holds for sub-gaussian distributions. See the proof of \cite[Prop B.5]{diakonikolas2018sever} for details. More technical comparisons with the concurrent papers \cite{diakonikolas2018sever,prasad2018robust} can be found in Remark \ref{rmk:gradient_approx}.

\vskip \baselineskip 

Let $S^{d-1}= \sth{v \in \reals^d: \norm{v}=1}$ denote the unit Euclidean sphere. 
\begin{assumption}\label{ass:gradient_sub_gaussian}
The sample gradient at the optimal model parameter $\theta^*$, i.e., $\nabla f(X, \theta^*)$, is sub-exponential with 
constants $(\sigma_1,\alpha_1)$, \ie, for every unit vector $v \in S^{d-1}$, 
\begin{align*}  
\expect{ \exp \left( \lambda \Iprod{ \nabla f(X, \theta^*) }{v} \right) } \le e^{\sigma_1^2 \lambda^2/2},
\quad \forall |\lambda| \le \frac{1}{\alpha_1}.
\end{align*}
\end{assumption}
We further assume the Lipschitz continuity of the sample gradient functions. 
\begin{assumption}\label{ass:bounded_hessian}
There exists an $L^{\prime}$ such that 
$$
\norm{\nabla  f(X,\theta) - \nabla f(X, \theta^{\prime})} ~ \le ~ L^{\prime} \norm{\theta - \theta^{\prime}} ~~ \forall ~ \theta, \theta^{\prime} \in \Theta. 
$$
\end{assumption}
For applications where Assumption \ref{ass:bounded_hessian} does not hold  deterministically, 
it suffices to have Assumption \ref{ass:bounded_hessian} hold with high probability 
for all training data. 
Notably, $L'$ may be
much larger than $L$ and even scale polynomially in $N$ and $d$.  However, $L'$ 
will affect our results only by logarithmic factors $\log L'$. 

Next define the gradient difference function
\begin{align}
\label{eq: t1}
h(X, \theta) ~ = ~ \nabla f(X, \theta) - \nabla f(X, \theta^*) - \left( \nabla F(\theta)- \nabla F(\theta^*) \right). 
\end{align}
Note that $h(X,\theta)/\norm{\theta-\theta^*}$ characterizes the change rate of $f(X,\theta) - \nabla F(\theta)$ from 
$f(X,\theta^*) - \nabla F(\theta^*)$; hence it can be viewed as
a local Lipschitz parameter with respect to $\theta^*$. 
\begin{assumption}\label{ass:gradient_sub_exp}
The local Lipschitz parameter $h(X, \theta)/\norm{\theta-\theta^*}$ is sub-exponential with constants $(\sigma_2, \alpha_2)$, \ie., for every  $\theta \in \Theta$ and $v \in S^{d-1}$, 
$$
\expect{\exp \left( \frac{  \lambda \Iprod{ h(X,\theta) }{v} }{\|\theta-\theta^*\|} \right)}
\le e^{\sigma_2^2 \lambda^2 /2},
\quad \forall |\lambda| \le \frac{1}{\alpha_2}. 
$$
\end{assumption} 
Notably, Assumption \ref{ass:gradient_sub_exp} assumes a concentration of 
the \emph{local} Lipschitz parameter with respect to $\theta^*$, instead of a {\em global} Lipschitz parameter. 

Again, our analysis may not be tied to sub-exponential assumptions. 
%
%
Now we are ready to present our main results. 
\begin{theorem}
\label{thm: opt}
Suppose Assumptions \ref{ass:pop_risk_smooth}--\ref{ass:gradient_sub_exp} hold.
Assume that $\log (L+L') =O(\log (Nd) )$ and $\Theta \subset \{\theta: \norm{\theta-\theta^*} \le r \}$
for some positive parameter $r$ such that $\log r = O(\log (Nd)  )$. 
Suppose that $N \ge c d^2 \log^8 (Nd))$ and $N \ge c q$ 
for a sufficiently large constant $c$, and that $4q \le m \le e^{\sqrt{d}}$.   
Further assume that $M\ge 1$. 
Then there exists a gradient aggregator $\calR$ 
such that  with probability at least $1-3e^{-\sqrt{d}}$,
the iterates $\{\theta_t\}$ given by Algorithm \ref{alg:BGD} satisfy 
$$
\norm{\theta_t - \theta^* } \lesssim  \pth{1-\frac{M^2}{16L^2}}^t \norm{\theta_0-\theta^*} +  \sqrt{\frac{q}{N}} + \sqrt{\frac{d }{N} } . 
$$
\end{theorem}
Note that in Theorem \ref{thm: opt}, the Lipschitz parameters $L, L^{\prime}$, and the size of the search space $r$ are allowed to scale even polynomially in $N$ and $d$. 
%
The estimation error $O(\sqrt{q/N} + \sqrt{d/N})$ in Theorem \ref{thm: opt} significantly improves the previous results ($O(\sqrt{dq/N})$ for $q\ge 1$) 
\cite{Chen:2017:DSM:3175501.3154503}.  
Recall that even in the failure-free and centralized setting the minimax-optimal error rate is $O(\sqrt{d/N})$. Thus, an immediate consequence of Theorem \ref{thm: opt} is that 
as long as $q=O(d)$, our proposed algorithm achieves the optimal error rate $O(\sqrt{d/N})$. 
The sample complexity $N \ge c d^2 \log^8 (Nd)$ appears to be suboptimal at first glance. However, it turns out that under sub-exponential distributional assumption, if we rely on uniform concentration of the sample covariance matrices, this sample complexity is order optimal up to logarithmic factors.
See Remark \ref{rmk: sample complexity} for details. 

\begin{remark}
\label{rmk: coordinate-wise median}
One concurrent work~\cite{yin2018byzantine} uses coordinate-wise median and trimmed mean
and obtains an estimation error 
$rd(  q/(m\sqrt{n}) + 1/\sqrt{N} )$ up to logarithmic factors -- noting that the radius  of the model parameter
space $r$ is typically on the order of $\sqrt{d}$. 
This error rate
is shown to be order-optimal in the 
low dimensional regime where $d=O(1)$~\cite{yin2018byzantine}, but turns out to scale poorly in
dimension $d$. 
\end{remark}

As an important ingredient of our proof of Theorem \ref{thm: opt}, we establish a {\em uniform} concentration of the sample covariance matrix of  gradients. Recall that in our problem local sample gradients $\frac{1}{n} \textstyle{\sum}_{i\in \calS_j} \nabla f (X_i, \theta)$'s are sub-exponential random vectors. 
Standard routine to bounding the 
spectral norm of the sample covariance matrix 
is available, 
 see \cite[Theorem 5.44]{vershynin2010introduction} and \cite[Corollary 3.8]{adamczak2010quantitative} for example. 
 However, it turns out that using these existing results, the uniform 
 concentration bound obtained is far from being optimal. To this end, 
 we develop a new concentration inequality for matrices with i.i.d.\ sub-exponential  column vectors. As can be seen later, this new inequality leads to a near-optimal uniform bound.

\begin{theorem}\label{thm:sub_exp_matrix}
Let $A$ be a $d \times m$ matrix whose columns $A_j$ are independent and identically distributed sub-exponential, 
zero-mean random vectors 
in $\reals^d$ with parameters $(\sigma,\alpha)$. Assume that 
\begin{align}
\label{eq: scale invariant}
\sigma/\alpha=\Omega(1).
\end{align}
Then with probability at least $1-\delta$,
$$
 \opnorm{A} \le c \left( \sigma\sqrt{m}  + \sigma \phi \left(  d + \log \frac{1}{\delta} \right)  + \alpha \phi^2\left(  d + \log \frac{1}{\delta} \right)  \right) ,
$$
where $c$ is a universal positive constant and $\phi(x): \reals \to \reals$ is a function given by
$
\phi(x) =  \sqrt{x} \log^{3/2} (x).
$
\end{theorem}
If $A$ has sub-Gaussian columns, \ie, $\alpha=0$, 
then the upper bound in 
\prettyref{thm:sub_exp_matrix} matches the sub-Gaussian matrix concentration inequality~\cite{vershynin2018}[Theorem 5.39] 
up to logarithmic factors. If $\sigma,\alpha=\Theta(1)$ and $\log(1/\delta) = d$, \prettyref{thm:sub_exp_matrix}
implies that with probability at least $1-e^{-d},$
$\norm{A} \lesssim  \sqrt{m} +  d \log^3 d$, 
which is tight up to logarithmic factors;
whereas the analogous bound implied by standard concentration inequality~\cite[Corollary 3.8]{adamczak2010quantitative}
is  on the order of  $\sqrt{md } + d.$ See Remark \ref{rmk: spectral bounds} in Section \ref{proof of thm matrix} for details. 

With the performance guarantee of our robust aggregator $\calR$ 
(formally stated in the next subsection) 
and Theorem \ref{thm:sub_exp_matrix}, it can be shown that 
the approximate gradients used in the robust gradient descent update 
(Algorithm \ref{alg:BGD} with the chosen robust aggregator $\calR$) are good uniformly over $\theta \in \Theta$. 
\begin{theorem}
\label{thm: uniform bound}
Suppose Assumptions \ref{ass:pop_risk_smooth}--\ref{ass:gradient_sub_exp} hold.
Assume that $\log (L+L') =O(\log (Nd) )$ and $\Theta \subset \{\theta: \norm{\theta-\theta^*} \le r \}$
for some positive parameter $r$ such that $\log r = O(\log (Nd)  )$. Suppose $N \ge c d^2 \log^8 (Nd))$ for a sufficiently large constant $c$ and $m \le e^{\sqrt{d}}$.
Let $G(\theta)$ (for each $\theta \in \Theta$) be the aggregated gradient returned by 
Algorithm \ref{ag: robust mean}.
Then with probability at least $1- 3e^{-\sqrt{d} }$,  
\begin{align*}
\norm{G(\theta)-\nabla F(\theta)} ~ \lesssim  ~ & \left(  \sqrt{\frac{q}{N} } + \sqrt{ \frac{d}{N} }  \log^2 (Nd)  \right) 
 \norm{\theta-\theta^*}  \\
 & +    \sqrt{\frac{q}{N}} + \sqrt{\frac{d }{N} } 
\end{align*}
holds for all $\theta \in \Theta$. 
\end{theorem}

\begin{remark}
\label{rmk: sample complexity}
Theorem \ref{thm: uniform bound} requires the total sample size $N \gtrsim d^2$ (ignoring the logarithmic factors), 
which is due to our sub-exponential assumption of local Lipschitz parameter $h(X,\theta)/\norm{\theta-\theta^*}$. This sample size requirement $N \gtrsim d^2$ is inevitable as can be seen from the linear regression example. 
(cf. \prettyref{rmk:tight_upper_bound}). 
If instead $h(X,\theta)/\norm{\theta-\theta^*}$ is assumed to be sub-Gaussian, then $N \gtrsim d$ suffices. 
\end{remark}

\begin{remark}\label{rmk:gradient_approx}
The robust gradient estimation has also been studied
in two concurrent papers~\cite{diakonikolas2018sever,prasad2018robust} for $m=N.$   
Under an $\epsilon$-contamination model with $\epsilon=q/N$, 
\cite[Lemma 1]{prasad2018robust} proves a \emph{point-wise} approximation bound to $\nabla F(\theta)$
for a given $\theta$,
which scales as $\sqrt{q/N} + (d/N)^{3/8} + q^{1/4} \sqrt{d/N}$ up to logarithmic factors. 
In contrast, a uniform approximation bound similar to ours is proved 
in \cite[Prop B.5]{diakonikolas2018sever}. 
 However, their bound only holds under
the stringent condition $q \gg d^2 (L')^2$,
where $L'$ is the Lipschitz continuity parameter of  the sample gradient function given in 
\prettyref{ass:bounded_hessian}. Note that $L'$ may scale polynomially in $N, d$. For instance,
in the standard linear regression model, $L'=\Omega(d)$ (See \prettyref{lmm:least-squares}).
In fact, \cite[Prop B.5]{diakonikolas2018sever} assumes that 
$\norm{\nabla f(X, \theta) - \nabla F(\theta)}$ is bounded by $L'$ 
and hence the proof follows from a straightforward application of 
Hoeffding's inequality plus an $\epsilon$-net argument. 

In contrast, with Assumption \ref{ass:gradient_sub_exp}, we obtain a tighter uniform approximation bound via the new matrix concentration inequality given in Theorem \ref{thm:sub_exp_matrix}. 
\end{remark}

\begin{remark}[Strongly-convex assumption]
\label{rmk:strongly-convex}
We next explain the strong convexity assumption assumed in our paper.
First, we clarify that we only require \emph{the population risk} to be strongly-convex, 
while the empirical risk (sample version) could be highly non-convex. 
Second, we point out that 
it is possible to enforce the strong convexity of the population risk by introducing proper regularization.
For example,  we add an extra $\ell_2$ norm regularization to the quadratic loss in ridge regression. 
Thirdly, while our results do not directly apply to settings where the population risk is not strongly convex, 
our results are still important  for the following two-fold reasons: (1) 
while many learning problems are highly non-convex globally, they sometimes
satisfy certain 
\emph{restricted strong convexity} properties (the Hessian matrix are strictly positive definite 
in certain regions or directions around the optimal model parameter) \cite{Negahban2011Restricted}
 and thus gradient descent
schemes are still able to converge to the optimal model parameter \cite{Ma17}. 
Therefore, our robust gradient descent results can still be applied to these settings; 
(2) 
It is possible to extend our results to non-convex settings
by combining our robust gradient descent methods with proper saddle-points escaping 
schemes; such extension has been pursued recently in a follow-up work~\cite{yin2018defending}.

\end{remark}

\begin{remark}\label{rmk:sub-exponential} 

Finally, we explain the sub-exponential assumption on the gradient vectors. 
First, while sub-gaussian \emph{data distribution} is commonly assumed in statistical learning literature,  
the resulting \emph{gradients} may be sub-exponential, as we illustrated in the simple linear regression example. 
Loosely speaking, this is due to the fact that 
the gradients may involve  $x^2$ term, which is sub-exponential even if $x$ is sub-Gaussian.
Similar phenomenon also occurs in logistic regression. 
Thus, it is important to consider sub-exponential or even heavier-tailed 
\emph{gradient distribution} when we try to robustify a learning procedure such as gradient descent. 
Second, our analysis framework developed in this work is not tied to sub-exponential assumptions. 
With different gradient distributional assumptions such as bounded second moment, 
one can follow our analysis roadmap (Lemma 2) to obtain
different uniform concentration bounds for the sample covariance matrix of gradient vectors which in turn
implies different error bounds to our robust gradient descent method. 
Thirdly, certain distributional assumptions are inevitable to some extent,
in order to show the learning procedures perform well on the average case beyond the worst case 
guarantees~\cite{vershynin2018}.

\end{remark}

\subsection{Robust Gradient Aggregator}
\label{subsec: robust agg}

In this subsection, we present the robust gradient aggregator $\calR$ used in Algorithm \ref{alg:BGD}. 
We present the aggregator in the setup of robust mean estimation. 

Let $\calS = \sth{y_1, \cdots, y_m}$ be the true sample. 
Define $\mu_{\calS} =\frac{1}{m} \sum_{i=1}^m y_i$ as the sample mean on $\calS$.
Let $\sth{\hat{y}_1, \cdots, \hat{y}_m} \subseteq  \reals^d$ be the observed sample, which is obtained from $\calS$ by adversarially corrupting up to $q =\epsilon m$ data points. 
We use an iterative filtering algorithm proposed in
\cite{Steinhardt18}, formally presented in Algorithm \ref{ag: robust mean}. 
At a high level, by solving \eqref{eqn: opt} and \eqref{eq:opt_max_min} for a saddle point $(W,U)$, 
Algorithm \ref{ag: robust mean} iteratively finds a 
direction (given by $U^*$) along which all data points are spread out the most, 
and filters away data points which have large residual errors
projected along this direction (given by \prettyref{eq:residual_error}). See Appendix \ref{app: robust mean} for detailed discussions.
\begin{algorithm}
\caption{Iterative Filtering for Robust Mean Estimation \cite{Steinhardt18}}
\label{ag: robust mean}
{\em Input}: Sample $\sth{\hat{y}_1, \cdots, \hat{y}_m}\subseteq \reals^d$, 
 $1-\alpha \triangleq \epsilon \in [0,\frac{1}{4})$, and $\sigma>0$. \\ 
{\em Initialization}: $\calA \gets \sth{1, \cdots, m}$, 
$c_i \gets 1$ and $\tau_i \gets 0$ for all $i\in \calA$. 
\vskip 0.2\baselineskip 
\begin{algorithmic}[1]
\WHILE{true}
\vskip 0.2\baselineskip 
 \STATE For $W \in \reals^{|\calA| \times |\calA|}$ and $U \in \reals^{d \times d}$, define a cost function 
$\psi: (W,U) \to \reals$ as:
\begin{align*}
\psi(W, U) = \sum_{i \in \calA}  c_i \big( \hat{y}_i - \sum_{j \in \calA} \hat{y}_j W_{ji} \big)^\top U  
\big( \hat{y}_i - \sum_{j \in \calA} \hat{y}_j W_{ji} \big).
\end{align*}
Let $W^*$ be a minimizer to the following convex program:
 \begin{align}  
 \min_{ \substack{ 0\le W_{ji} \le \frac{4-\alpha}{\alpha (2+\alpha)m} \\ \sum_{j\in \calA} W_{ji}=1} }  
 ~~\max_{ \substack{U \succeq 0\\ \Tr(U) \le 1} }  ~~ \psi(W, U)
 \label{eqn: opt}
 \end{align}
 and  $U^*$ be a maximizer to the following convex program:
 \begin{align}
  \max_{ \substack{U \succeq 0 \\ \Tr(U) \le 1 }} ~~ \min_{ \substack{ 0\le W_{ji}\le \frac{4-\alpha}{\alpha (2+\alpha)m} \\ \sum_{j\in \calA} W_{ji}=1} } ~~ \psi(W, U)
  \label{eq:opt_max_min}
 \end{align}
\STATE For $i \in \calA$, 
\begin{align}
\tau_i  \gets \big( \hat{y}_i - \sum_{j \in \calA} \hat{y}_j W^*_{ji} \big)^\top U^*  \big( \hat{y}_i - \sum_{j \in \calA} \hat{y}_j W^*_{ji} \big). \label{eq:residual_error}
\end{align}
\IF{$\sum_{i\in \calA} c_i \tau_i > 8 m\sigma^2$}
\STATE For $i\in \calA$, 
$
c_i \gets \pth{1-\frac{\tau_i}{\tau_{\max}}} c_i,
$ 
where $\tau_{\max} =\max_{i\in \calA} \tau_i.$ 
\STATE $\calA \gets \calA /\sth{i: ~c_i\le \frac{1}{2}}$.   
\ELSE 
\STATE Break {\bf while}--loop. 
\ENDIF
\ENDWHILE
\vskip 0.2\baselineskip 
\RETURN  $\hat{\mu} = \frac{1}{|\calA|} \sum_{i\in \calA} \hat{y}_i$. 
\end{algorithmic}
\end{algorithm} 
%

Given the corrupted sample $\sth{\hat{y}_1, \cdots, \hat{y}_m}$, $\epsilon$, and $\sigma$, Algorithm \ref{ag: robust mean} {\em deterministically}  outputs an estimate $\hat{\mu}$ that differs from the true {\em sample mean} 
by at most a bounded distance, formally stated in Lemma \ref{lmm:algo}. 
 
\begin{lemma} \cite{Steinhardt18}
\label{lmm:algo}
Suppose that 
\begin{align}
\label{eq: bounded sample spectral}
\opnorm{\frac{1}{m} \sum_{i\in \calS} \pth{y_i - \mu_{\calS}}\pth{y_i - \mu_{\calS}}^{\top} } \le \sigma^2. 
\end{align}
Then for $q/m = \epsilon \le \frac{1}{4}$, Algorithm \ref{ag: robust mean} outputs a parameter $\hat{\mu}$ such that 
\begin{align}
\norm{ \hat{\mu}- \mu_{\calS}}  = O(\sigma \sqrt{\epsilon}).  \label{eq:error_bound_main}
\end{align}
\end{lemma}

Condition \eqref{eq: bounded sample spectral} 
ensures that the uncorrupted data points $y_i$'s
are well concentrated around the sample mean $\mu_{\calS}$
in every direction. If there are large residual errors found by 
Step 3 of Algorithm \ref{ag: robust mean}, they are likely caused by the corrupted data points rather than the good data points. 
%

\begin{remark}
Note that condition \eqref{eq: bounded sample spectral} is slightly different from that 
in \cite{Steinhardt18}: the summation is taken over the entire true sample $\calS$ rather than a subset of sample. We make this modification in order to include the regime  $q/m = o(1)$. 
For completeness, we present the proof of Lemma \ref{lmm:algo} in Appendix \ref{app: robust mean}. 
 
Also, the termination of Algorithm \ref{ag: robust mean} requires the knowledge of $\sigma$; however, in the setup of statistical learning, this might further call for the knowledge of $\theta^*$, which is not practical. Considering this, we have an alternative termination condition which does not need to know any parameter other than $\epsilon$. 
See Appendix \ref{app: ag modification} for details. 
\end{remark}

Formally, we use Algorithm \ref{ag: robust mean} as our robust gradient aggregator $\calR$ with 
inputs 
$$
\hat{y}_1(\theta) = g_1(\theta), ~ \cdots,  ~  \hat{y}_m(\theta) = g_m(\theta),
$$
where $g_1(\theta), \cdots,  g_m(\theta)$ are the local gradient functions computed by the $m$ workers, 
among which up to $q$ reported gradient functions may not be the true local gradient functions.  
The true $m$ local gradient functions are given by 
 \begin{align} 
 y_1(\theta)=  \textstyle{\frac{1}{n} \sum}_{i\in \calS_1} \nabla f (X_i, \theta), ~ \cdots, 
 ~y_m(\theta) = \textstyle{\frac{1}{n}  \sum}_{i\in \calS_m} \nabla f (X_i, \theta).
 \label{eq:def_y}
\end{align}
 -- recalling that $| \calS_j|=n=\frac{N}{m}$ for each $j\in [m]$. 
 The true sample mean $\mu_{\calS}(\theta)$ is 
\begin{align}
\mu_{\calS}(\theta) = \frac{1}{m}\sum_{j=1}^m y_j(\theta) = \frac{1}{N} \sum_{i=1}^N \nabla f (X_i, \theta),
\label{eq:def_sample_mean}
\end{align}
and the population mean $\mu(\theta)$ is $\nabla F(\theta).$

Note that other robust mean estimation algorithms given in \cite{DiakonikolasKK017, lai2016agnostic} might also suffice for our purpose, and we would like to explore these different robust gradient aggregation schemes in the future. 

\section{Main Analysis}


Before presenting our main analysis,  we briefly discuss two important implications of \prettyref{lmm:algo} in the robust mean estimation setup.  

\begin{remark}\label{remark: sample to population mean}
In \prettyref{lmm:algo}, the estimation error bound \prettyref{eq:error_bound_main} is in 
terms of $\norm{\hat{\mu}-\mu_{\calS}}$. Let $\mu$ be the true mean of the unknown underlying distribution.
We can easily deduce an estimation error bound in terms of $\norm{ \hat{\mu}- \mu }$ from
the following triangle inequality  
 \begin{align*}
 \norm{ \hat{\mu}- \mu } &\le \norm{ \hat{\mu}- \mu_{\calS}} + \norm{ \mu_{\calS} - \mu} = O\left(\sigma \sqrt{\epsilon} \right) +\norm{ \mu_{\calS} -  \mu}. 
 \end{align*}
  Thus, to characterize the estimation error $\norm{ \hat{\mu}-  \mu}$, it is enough to control the spectral norm of the \emph{true} sample covariance matrix $\opnorm{\frac{1}{m} \sum_{i\in \calS} \pth{y_i - \mu_{\calS}}\pth{y_i - \mu_{\calS}}^{\top} }$ 
 and the deviation of the empirical average $\norm{ \mu_{\calS} - \mu}$  -- the latter of which is standard. 
\end{remark}

\begin{remark}\label{rmk: spectral triangle}
Note that 
 \begin{align*}
 & \opnorm{\frac{1}{m} \sum_{i\in \calS} \pth{y_i - \mu_{\calS}}\pth{y_i - \mu_{\calS}}^{\top}} \\
 	& = \frac{1}{m} \opnorm{\pth{\qth{y_1, \cdots, y_m} - \mu_{\calS} \bm{1}_m^{\top} }\pth{\qth{y_1, \cdots, y_m} - \mu_{\calS} \bm{1}_m^T}^{\top}}\\
 & = \frac{1}{m} \opnorm{\qth{y_1, \cdots, y_m} - \mu_{\calS} \bm{1}_m^{\top}}^2\\
 & \le \frac{1}{m} \left( \opnorm{\qth{y_1, \cdots, y_m} - \mu \bm{1}_m^{\top} }+ \sqrt{m}\norm{\mu - \mu_{\calS}}\right)^2,
 \end{align*}
 where $\bm{1}_m \in \reals^m$ is an all-ones vector. Therefore, to derive $\sigma$ in condition \eqref{eq: bounded sample spectral},
 it is enough to bound $\norm{\mu - \mu_{\calS}}$ and $\frac{1}{\sqrt{m}} \opnorm{\qth{y_1, \cdots, y_m} -  \mu \bm{1}_m^{\top}}$.
\end{remark}

Back to our statistical learning problem, as discussed in Remarks \ref{remark: sample to population mean} and \ref{rmk: spectral triangle},  
to guarantee that the aggregated gradient is close to the true gradient uniformly over all $\theta \in \Theta$,
it suffices to bound 
\begin{align}
\label{eq: uniform 2}
\norm{ \textstyle{\frac{1}{N} \sum_{i=1}^N} \nabla f (X_i, \theta) -\nabla F(\theta)}
\end{align}
and
\begin{align}
\label{eq: uniform 1}
\frac{1}{\sqrt{m}} 
 \big\| \big[  
 & \textstyle{\frac{1}{n} \sum_{i\in \calS_1}} \nabla f (X_i, \theta) -\nabla F(\theta), ~ \cdots~, \nonumber \\
&  \textstyle{\frac{1}{n} \sum_{i\in \calS_m}} \nabla f (X_i, \theta) - \nabla F(\theta) 
\big] \big\| 
\end{align}
uniformly over all $\theta \in \Theta$. Getting a uniform bound to 
\prettyref{eq: uniform 2} involves  
standard concentration of sum of i.i.d.\ random vectors and is relatively easy. 
The main challenge is to uniformly bound \prettyref{eq: uniform 1}, for which
we develop a (nearly tight) matrix concentration inequality.

\subsection{New Matrix Concentration Inequality: \prettyref{thm:sub_exp_matrix} and its Proof}
\label{proof of thm matrix}

The aim of this subsection is to present our main technical tool to derive a tight uniform
bound to \prettyref{eq: uniform 1}. This subsection is independent from the rest of the
paper and can be skipped at the first reading. 

For any fixed $\theta$, the matrix in \prettyref{eq: uniform 1}
is of independent columns; standard routine to bound \prettyref{eq: uniform 1} point-wise 
is available, see \cite[Theorem 5.44]{vershynin2010introduction} and \cite[Corollary 3.8]{adamczak2010quantitative} for example. To get a uniform concentration result, we can use $\epsilon$--net argument to extend the concentration of a fixed $\theta$ to uniform over all $\theta \in \Theta$. Nevertheless, using these standard matrix concentration results, the 
uniform concentration bound obtained is far from being optimal, as we explain next. 

The following theorem is, to the best of our knowledge, a state-of-the-art concentration inequality for matrices with sub-exponential columns~\cite[Corollary 3.8]{adamczak2010quantitative}.
\begin{theorem}\label{thm:sub_exp_spectral_2} 
Let $A$ be a $d \times m$ matrix whose columns $A_j$ are i.i.d., zero-mean, 
sub-exponential random vectors in $\reals^d$ with the  scaling parameters $\sigma$ and $\alpha$. 
Assume that $\sigma,\alpha=O(1)$ and $m \le e^{\sqrt{d}}$. 
There are absolute positive constants $C$ and $c$ such that for every $K \ge 1$, with
probability at least $1-e^{-c K \sqrt{d}},$
$$
\opnorm{A} \le C K \left( \sqrt{m} + \sqrt{d} \right). 
$$
\end{theorem}
Note that assuming $m \le e^{\sqrt{d}}$ only loses minimal generality in the high-dimensional regime. 
The above theorem is tight up to constant factors when the tail probability is  
on the order of $e^{-\sqrt{d}}$, i.e., $K=\Theta(1)$, see~\cite[Remark 3.7]{adamczak2010quantitative}
for a proof. 
However, in our problem, to guarantee a uniform bound to \prettyref{eq: uniform 1}
via an $\epsilon$--net argument, 
we need a tail probability on the order of $e^{-d}$, i.e., $K\approx \sqrt{d}$. 
In this case, \prettyref{thm:sub_exp_spectral_2} yields an upper
bound on the order of $\sqrt{md} + d$. 
Using matrix Bernstein's inequality given by
\cite[Theorem 5.44]{vershynin2010introduction} instead, we can obtain an alternative upper bound 
$O(\sqrt{m}+d^{3/2})$. 
Both of these two upper bounds turn out to be loose. To this end, we develop a new matrix concentration 
inequality, proving a nearly-tight upper bound on the order of $\sqrt{m}+d$ up to
logarithmic factors.

A key step in deriving a concentration inequality for matrices with sub-exponential random vectors
is to obtain a large deviation inequality for 
the sum of independent random variables whose tails decay 
\emph{slower} than sub-exponential random variables. 
Note that in this case, the moment generating function may not exist 
and thus we cannot follow the standard approach to obtain a large deviation
inequality by invoking the Chernoff bound. To circumvent this,  
we partition the support of a real-valued random variable $Y$ into countably many finite segments, and write $Y$ as a summation of component random variables, each of which is supported on its corresponding segment. Due to the fact that each segment is of finite length, we can apply Bennett's inequality for bounded random variables (cf. \prettyref{lmm:Bennett}).
Then we take a union bound to arrive at a concentration result of the original $Y$. Some additional care is needed in choosing the partition. Our proof is inspired by
Proposition 2.1.9 and Excercise 2.1.7 in~\cite{tao.rmt}.

\begin{lemma}
\label{lmm:truncation_con}
Let $Y$ be a random variable whose tail probability satisfies 
\begin{align*}
\prob{\abth{Y} \ge t} \le \exp \pth{-E(t)},
\end{align*}
where $E(t): \reals_+ \to [0, \infty]$ is a non-decreasing
function. 
Suppose that 
\begin{align}
\label{ass: tail 2}
E(t)/t \text{ is monotone in } t,
\end{align}
and there exists $t_0 \ge e^2$ such that for all $t \ge t_0$ and
all $k$ with $4(k+1)^2 e^{k} \ge t$,
\begin{align}
\label{eq: tail condition}
 E(e^{k-1}) \ge  2\pth{2k+ 4\log(k+1) + \log 2-\log t}.
\end{align}
Let $Y_1, \cdots, Y_m$ be $m$ independent copies of $Y$. 
If $E(t)/t $ is non-decreasing, then
\begin{align}
& \prob{\abth{\textstyle{\sum}_{j=1}^m Y_j - m\expect{Y}} \ge mt}  \nonumber  \\
&\le 2 \log (mt)\exp \pth{- \frac{ m }{4( \log (mt)+1)^2 } E\pth{ \frac{t}{4e \log^2 t } } } \nonumber \\
&\quad + \exp \pth{-\frac{1}{2}E\pth{\frac{mt}{e}}}; \label{lm: con increase}
\end{align}
if $E(t)/t $ is non-increasing, then
\begin{align}
& \prob{\abth{\textstyle{\sum}_{j=1}^m Y_j - m\expect{Y}} \ge mt}  \nonumber  \\
& \le 2 \log (mt)\exp \pth{- \frac{1}{4 e(\log (mt)+1)^2} E\pth{\frac{mt}{e}}} \nonumber \\
& \quad +  \exp \pth{-\frac{1}{2}E\pth{\frac{mt}{e}}}. \label{lm: con decrease}
\end{align} 
\end{lemma} 

\begin{remark}
To illustrate the upper bound \eqref{lm: con increase},
let us consider the following special cases.

Case 1: Suppose $Y$ is sub-Gaussian. In this case, $E(t) = c t^2$ for a universal constant $c>0$. 
Thus, there exists a universal constant $t_0\ge e^2$ such that
	both \eqref{eq: tail condition} and \eqref{ass: tail 2} hold. 
	Then \eqref{lm: con increase} gives the desired sub-Gaussian tail bound 
	$e^{-\Omega(mt^2)}$ up to logarithmic factors in the exponent. 
	
Case 2: Suppose $Y$ is sub-exponential. In this case, $E(t)= c t$ for a universal constant $c$. Thus, 
	there exists $t_0\ge e^2$ that only depends on $c$ such that 
	both \eqref{eq: tail condition} and \eqref{ass: tail 2} hold.
	Then \eqref{lm: con increase} gives 
	the desired sub-exponential tail bound $e^{-\Omega(mt)}$ 
	up to logarithmic factors in the exponent. 

Case 3: Suppose $Y=Z^2$, where $Z$ is sub-exponential. 
In this case, $E(t)= c \sqrt{t}$ for a universal constant $c>0$.
	Thus, there exists $t_0\ge e^2$ that only depends on $c$ such that 
	both \eqref{eq: tail condition} and \eqref{ass: tail 2} hold. 
Then \eqref{lm: con decrease} gives 
	a tail bound $e^{-\Omega(\sqrt{mt})}$ 
	up to logarithmic factors in the exponent.
\end{remark}
Despite the fact that Lemma \ref{lmm:truncation_con} is loose up to logarithmic factors compared to the standard sub-gaussian and sub-exponential random variables, Lemma \ref{lmm:truncation_con} applies to much larger family than the sub-gaussian distributions, and requires much less structure on the distributions. In particular, Lemma \ref{lmm:truncation_con} does not require the existence of moment generating function. 

The proof of Lemma \ref{lmm:truncation_con} can be found in Appendix \ref{app: truncation}. 
\prettyref{lmm:truncation_con} is our key machinery to obtain the concentration inequality
for matrices with i.i.d.\ sub-exponential  random vectors given in \prettyref{thm:sub_exp_matrix}.
We restate the theorem below for ease of reference.

\begin{theorem*}[\prettyref{thm:sub_exp_matrix}]
\label{thm:sub_exp_matrix_restate}
Let $A$ be a $d \times m$ matrix whose columns $A_j$ are independent and identically distributed sub-exponential, 
zero-mean random vectors 
in $\reals^d$ with parameters $(\sigma,\alpha)$. Assume that 
\begin{align}
\sigma/\alpha=\Omega(1).
\end{align}
Then with probability at least $1-\delta$,
$$
 \opnorm{A} \le c \left( \sigma\sqrt{m}  + \sigma \phi \left(  d + \log \frac{1}{\delta} \right)  + \alpha \phi^2\left(  d + \log \frac{1}{\delta} \right)  \right) ,
$$
where $c$ is a universal positive constant and $\phi(x): \reals \to \reals$ is a function given by
$
\phi(x) =  \sqrt{x} \log^{3/2} (x).
$
\end{theorem*}
\begin{remark}
\label{rmk: spectral bounds}
We discuss two consequences of \prettyref{thm:sub_exp_matrix}.

	Suppose $\alpha=0$. In this case, $A$ has sub-Gaussian columns, and \prettyref{thm:sub_exp_matrix} 
implies that 
$$
\opnorm{A} ~ \lesssim ~ \sigma \left( \sqrt{m} + \sqrt{d+ \log \frac{1}{\delta}} \log^{3/2} \left( d + \log \frac{1}{\delta} \right) \right),
$$
which matches the sub-Gaussian matrix concentration inequality~\cite[Theorem 5.39]{vershynin2018} 
up to logarithmic factors. 

Suppose $\sigma,\alpha=\Theta(1)$, and $\log(1/\delta) = d$. 
In this case, we get  that
with probability at least $1-e^{-d},$
\begin{align}
\norm{A} ~ \lesssim ~ \sqrt{m} +  d \log^3 d  \quad \text{ implied by  \prettyref{thm:sub_exp_matrix} },  \label{eq:norm_bound_exp}
\end{align}
whereas the analogous bound implied by \prettyref{thm:sub_exp_spectral_2} is on the order of  $\sqrt{md } + d.$
Using matrix Bernstein's inequality given by
\cite[Theorem 5.44]{vershynin2010introduction} instead, we can obtain an alternative upper bound 
$O(\sqrt{m}+d^{3/2})$.
The upper bound \prettyref{eq:norm_bound_exp} is tight up to logarithmic factors; see \prettyref{app:tightness_upper_bound} for a proof. 
\end{remark}

The proof of \prettyref{thm:sub_exp_matrix} also uses  the 
following standard concentration inequality for sum of independent sub-exponential random variables. In particular, we use this concentration inequality to get a concentration bound at $\theta^*$.  
\begin{lemma}\cite[Proposition 2.2]{wainwright2015basic}
\label{lmm:sum_sub_exponential}
Let $Y_{1}, \ldots, Y_{m} $ denote a sequence of independent random variables, where $Y_{j}$'s are sub-exponential
with scaling parameters $(\sigma_{j}, \alpha_{j})$ and mean $0$. Then $\sum_{j=1}^{m} Y_{j}$ is sub-exponential with scaling parameters $(\sigma,\alpha)$, where $\sigma^{2} = \sum_{j=1}^{m}\sigma_{j}^{2}$ and $\alpha = \max_{1\leq j \leq m}\alpha_{j}.$ Moreover,
\begin{align*}
\prob{ \sum_{j=1}^{m} Y_j \geq t } \leq
	\begin{cases}
		\exp  \left( - \frac{t^{2}}{  2\sigma^{2} } \right)  & \text{ if }  0 \le t \le \sigma ^{2}/\alpha; \\
		\exp \left(  - \frac{t}{  2\alpha  } \right) &  \text{ o.w. }	
		\end{cases}
\end{align*}
\end{lemma}

The following lemma gives an upper bound to the spectral norm of the covariance matrix  
 of a sub-exponential random vector.
\begin{lemma}\label{lmm:sub_exp_covariance}
Let $Y \in \reals^d$ denote a zero-mean, sub-exponential random vector with scaling parameters $(\sigma,\alpha)$, and
$\Sigma$ denote its covariance matrix $\Sigma=\expect{YY^\top}.$ Then
$$
\opnorm{\Sigma} \le 4 \sigma^2 + 16 \alpha^2.
$$
\end{lemma}
\begin{proof}
First recall that 
\begin{align*}
\opnorm{\Sigma} ~ &= \sup_{v\in S^{d-1}} v^\top \Sigma  v
= \sup_{v\in S^{d-1}} v^\top \expect{YY^\top} v 
= \sup_{v\in S^{d-1}} \expect{\iprod{Y}{v}^2}.
\end{align*}
For each unit vector $v$, from \cite[Exercise 1.2.3]{vershynin2018}, we have 
\begin{align}
\label{eq: exponential second moment}
\nonumber
\expect{\iprod{Y}{v}^2} 
&  = \int_{0}^{\infty} 2t ~ \prob{\abth{\iprod{Y}{v}}\ge t}dt\\
\nonumber
& \le   \int_{0}^{\infty} 4t ~ 
\exp \left(  - \frac{1}{2} \min \sth{\frac{t^2}{\sigma^2}, \frac{t}{\alpha}} \right) dt \\
& \le 4 \sigma^2 + 16 \alpha^2. 
\end{align}
Note that the above upper bound is independent of $v.$ The lemma follows by combining the
last two displayed equations. 
\end{proof}

Now, we are ready to present the proof of Theorem \ref{thm:sub_exp_matrix}.
\begin{proof}[\bf Proof of Theorem \ref{thm:sub_exp_matrix}]
Recall $\Sigma=\expect{A_1A_1^\top}$. Then
$$
\opnorm{A}^2=\opnorm{A A^\top} \le \opnorm{AA^\top - m\Sigma} + m \opnorm{\Sigma}.
$$
In view of \prettyref{lmm:sub_exp_covariance}, we have $\opnorm{\Sigma} \le 4\sigma^2 + 16\alpha^2$.
It remains to bound $\opnorm{AA^\top -m \Sigma}.$
Note that  
\begin{align*}
\opnorm{A A^\top - m \Sigma} & = \sup_{v \in S^{d-1}} \abth{v^{\top } \left( A A^{\top} - m \Sigma \right) v} \\
&=\sup_{v \in S^{d-1}} \abth{\sum_{j=1}^m \left( \iprod{A_j}{v}^2 - \expect{\iprod{A_j}{v}^2} \right)}.
\end{align*}
Fix a $v \in S^{d-1}$. Note that 
$\iprod{ A_j}{v}$ is zero-mean sub-exponential random variable with parameter $(\sigma,\alpha)$.
For $j=1, \cdots, m$, define 
\begin{align}
\label{df: yj}
Y_j= \iprod{ A_j}{v}^2/\sigma^2.
\end{align}
It follows from~\prettyref{lmm:sum_sub_exponential} that
\begin{align*}
\prob{|Y_j | \ge t}  = \prob{ \abth{\iprod{A_j}{v}  } \ge \sigma \sqrt{t} }  
 \le 2\exp \left( - \min \left\{  \frac{t}{2} , \frac{\sigma \sqrt{t}}{2\alpha} \right\}\right),
\end{align*} 
We apply \prettyref{lmm:truncation_con} to $Y_1, \cdots, Y_m$ with 
$$
E(t) = \min \left\{  \frac{t}{2} , \frac{\sigma \sqrt{t}}{2\alpha} \right\}   - \log 2,
$$
which is non-decreasing in $t$. By assumption $\sigma/\alpha=\Omega(1)$, 
it follows that $E(t)$ scales as $\sqrt{t}$ in $t$. 
Thus there exists $t_0\ge e^2$ such that \eqref{eq: tail condition} holds. 
In addition, $E(t)/t$ is non-increasing. 
Therefore, \eqref{lm: con decrease} in Lemma \ref{lmm:truncation_con} applies, i.e., for all $t\ge t_0$, 
\begin{align}
\label{eq: concentration y}
\nonumber
& \prob{\abth{\sum_{j=1}^m \left( Y_j - \expect{Y_j} \right) } \ge mt} \\
\nonumber 
 & \le 
 2  \log (mt ) \exp \pth{- \frac{1}{4 e \log^2 (e mt) } E\pth{\frac{mt}{e}}}   + \exp \pth{-\frac{1}{2} E\pth{\frac{mt}{e}}}\\
 & \le 4  \log (mt ) \exp \pth{- \frac{1}{4 e \log^2 (e mt) } E\pth{\frac{mt}{e}}}. 
\end{align}

Next, we apply $\epsilon$-net argument. 
Let $\calN_{\frac{1}{4}}$ be the $\frac{1}{4}$--net of the unit sphere $S^{d - 1}$. 
From \cite[Lemma 5.2]{vershynin2010introduction}, we know that 
$
\abth{\calN_{\frac{1}{4}}} \le 9^d. 
$
In addition, it follows from \cite[Lemma 5.4]{vershynin2010introduction} that 
\begin{align*}
\opnorm{AA^\top - \Sigma} &   \le  
2 \sup_{v \in \calN_{\frac{1}{4}} } 
\abth{\sum_{j=1}^m  \left( \iprod{ A_j}{v}^2 - \expect{\iprod{ A_j}{v}^2} \right)}.
\end{align*} 
Hence, 
\begin{align*}
& \prob{\opnorm{AA^\top - \Sigma} \ge 2\sigma^2 m t }  \\
& \le \prob{ \sup_{v \in \calN_{\frac{1}{4}} } 
\abth{\sum_{j=1}^m \left( \iprod{ A_j}{v}^2 -  \expect{\iprod{ A_j}{v}^2 } \right)} \ge \sigma^2 m t}\\
& \le \abth{\calN_{\frac{1}{4}}} \prob{
\abth{\sum_{j=1}^m  \left( \iprod{ A_j}{v}^2 -  \expect{\iprod{ A_j}{v}^2 } \right)}  \ge \sigma^2 m t}\\
& \le 9^d \prob{ \abth{\sum_{j=1}^m \left( Y_j - \expect{Y_j}  \right) }  \ge  m t} ~~~ \text{by definition of  $Y_j$ in~\eqref{df: yj}}\\
& \le  \exp \pth{- \frac{1}{4 e \log^2 (e mt) } E\pth{\frac{mt}{e}} + \log 4 + \log \log (mt)+ d \log 9},
\end{align*}
where the last inequality holds by \eqref{eq: concentration y}. 
To complete the proof, we need to choose $mt$ so that 
the right hand side of the last inequality is smaller than $\delta$.
In other words,  we need to find $x \ge m t_0$ such that
$$
\frac{1}{4e \log^2(ex) }E(x/e) - \log \log x \ge \log \frac{4}{\delta} + d \log 9 \triangleq a.  
$$
One such $x$ is given by 
$$
x = c \left( a \log^3 a + \frac{\alpha^2}{\sigma^2} a^2 \log^6 a  +m \right),
$$
where $c$ is a sufficiently large constant. Therefore, we choose 
\begin{align*}
m t  = c  \bigg( & 
\left(d + \log \frac{1}{\delta} \right) \log^{3} 
\left( d + \log \frac{1}{\delta} \right)   \\ 
& + \frac{\alpha^2}{\sigma^2} \left(d + \log \frac{1}{\delta} \right)^2 
\log^{6} \left(d + \log \frac{1}{\delta} \right) + m \bigg).
\end{align*}
The theorem follows by taking the square root of $mt.$
\end{proof}

\subsection{Proof of Theorem \ref{thm: uniform bound}}
\label{sec: proof of approximate gradient}
With \prettyref{lmm:algo} and \prettyref{thm:sub_exp_matrix}, 
we are ready to prove Theorem \ref{thm: uniform bound}.
Recall that we need to bound \prettyref{eq: uniform 2}
and \prettyref{eq: uniform 1} uniformly for all $\theta \in \Theta$.
Bounding \prettyref{eq: uniform 2} uniformly is relatively easy 
and has been done in previous work~\cite[Proposition 3.8]{Chen:2017:DSM:3175501.3154503}.

\begin{proposition}\cite[Proposition 3.8]{Chen:2017:DSM:3175501.3154503}
\label{prop:uniform_converg}
Consider the same setup as \prettyref{thm: uniform bound}.  Assume that $N =\Omega( d \log (Nd) )$. 
Then with probability at least $1- e^{-d } $, 
\begin{align*}
\norm{ \frac{1}{N} \sum_{i=1}^N  \nabla f (X_i, \theta)-\nabla F(\theta) }
~ \lesssim ~  \Delta_2 \norm{ \theta-\theta^*}+  \Delta_1 , ~~~ \forall ~ \theta \in \Theta, 
\end{align*}
where 
$$
\Delta_1 \triangleq \sqrt{\frac{d }{N} },   \quad \text{and} \quad   \Delta_2 \triangleq \sqrt{  \frac{ d \log (Nd) }{N} }.
$$
\end{proposition}
It remains to bound \prettyref{eq: uniform 1}
uniformly over all $\theta \in \Theta$. 
For notational convenience, let 
\begin{align}
 G(X_{\calS}, \theta)  \triangleq 
\frac{1}{\sqrt{m}}  
\big[ & \textstyle{\frac{1}{n} \sum_{i\in \calS_1}} \nabla f (X_i, \theta) - \nabla F(\theta), ~ \cdots ~ , \nonumber \\
  &  \textstyle{\frac{1}{n} \sum_{i\in \calS_m}} \nabla f (X_i, \theta) - \nabla F(\theta) 
  \big].
  \label{dfn: R1}
\end{align}

\begin{proposition}
\label{prop: second mom 1}
Consider the same setup as \prettyref{thm: uniform bound}.  
With probability at least $1- 2e^{-\sqrt{d} }$,
for all $\theta \in \Theta$ 
\begin{align}
& \opnorm{G(X_{\calS}, \theta^\ast)} \lesssim \Delta_3 \\
& \opnorm{ G( X_{\calS}, \theta) - G( X_{\calS}, \theta^*)}
\lesssim \Delta_4 \norm{\theta-\theta^*} + \frac{1}{\sqrt{n} }, \label{eq:bound_tight}
\end{align}
where 
\begin{align*}
\Delta_3 ~ &\triangleq ~  \frac{1}{\sqrt{n} }  + \sqrt{ \frac{d}{N} } ,  \\
\Delta_4 ~ &\triangleq ~  \frac{1}{\sqrt{n} } 
+ \frac{1}{\sqrt{N} }   \phi \left(  2d +  d \log \pth{ 1+ r  \sqrt{n} (L+L')} \right) \\
 & \qquad + \frac{1}{\sqrt{Nn} }  \phi^2 \left(2d + d\log \pth{1+ r  \sqrt{n} (L+L')} \right),
\end{align*}
and $
\phi(x) =  \sqrt{x} \log^{3/2} (x).
$
It follows from triangle inequality that 
\begin{align*}
\opnorm{G(X_{\calS}, \theta)} ~ \lesssim ~   \Delta_4 \norm{\theta-\theta^*} + \Delta_3, \; \forall \theta \in \Theta.
\end{align*}
\end{proposition}
\begin{remark}\label{rmk:tight_upper_bound}
The uniform upper bound $\Delta_4$ in \prettyref{eq:bound_tight}
depends linearly in $d$ (ignoring logarithmic factors).
Such linear dependency is inevitable as can be seen from the standard linear regression model given
in \prettyref{sec: linear regression}. In this setting, 
$\nabla f(X_i, \theta)=w_i w_i^\top (\theta-\theta^*) - w_i \zeta_i$
and $\nabla F(\theta)=\theta-\theta^*$, where $w_i \iiddistr \calN(0, \identity)$
and $\zeta_i \iiddistr \calN(0,1)$ independent of $w_i$'s. 
For simplicity, assume $n=1$ and $m=N$.  Then
\begin{align*}
& \sup_{\theta \in S^{d-1}} \opnorm{ G(X_{\calS}, \theta) - G( X_{\calS}, \theta^*)}  \\
& \ge \sup_{\theta \in S^{d-1}} \frac{1}{\sqrt{N}}  
\norm{  ( w_1 w_1^\top  -\identity) (\theta-\theta^*)  }   \\
& = \frac{1}{\sqrt{N}} \left( \norm{w_1}^2 -1 \right) \norm{\theta-\theta^*} \\
& = O_P\left( \frac{d}{\sqrt{N}}\right) \norm{\theta-\theta^*},
\end{align*}
where $O_P\left( \frac{d}{\sqrt{N}}\right) \norm{\theta-\theta^*}$ denotes $O\left( \frac{d}{\sqrt{N}}\right) \norm{\theta-\theta^*}$ holds with high probability.  
The first equality follows by  
choosing $\theta-\theta^*$ parallel to $w_1$,
and the last equality holds by the concentration of $\chi^2$ distribution. 
\end{remark}

\begin{proof}[\bf Proof of Proposition \ref{prop: second mom 1}]

We prove the two bounds in \prettyref{eq:bound_tight} individually. 
\paragraph{Bounding $\opnorm{ G( X_{\calS}, \theta^*) }$:}
It follows from \prettyref{ass:gradient_sub_gaussian} that the columns of $G(X_{\calS}, \theta^*)$ are i.i.d.\ sub-exponential random vectors in $\reals^d$ with mean $0$ and
scaling parameters  $\sigma_1/\sqrt{nm}$ and $\alpha_1/(n\sqrt{m})$, where
$\sigma_1$ and $\alpha_1$ are two absolute constants.  
Therefore, the columns of the scaled matrix $\sqrt{N} \, G( X_{\calS}, \theta^*)$ 
are i.i.d.\ sub-exponential random vectors $\reals^d$ with mean $0$ and
scaling parameters $\sigma_1$ and $\alpha_1/\sqrt{n}$ -- recalling that $N=nm$. 
Applying \prettyref{thm:sub_exp_spectral_2} to $A= \sqrt{N} \, G( X_{\calS}, \theta^*)$, we get 
that with probability at least $1-e^{-\sqrt{d}},$
\begin{align}
\label{en: first part}
\opnorm{ G( X_{\calS}, \theta^*) } = \frac{1}{\sqrt{N}} \opnorm{A}  ~ \lesssim ~  \frac{1}{\sqrt{N}} \pth{\sqrt{m} + \sqrt{d}} =\frac{1}{\sqrt{n} }  + \sqrt{ \frac{d}{N} }. 
\end{align}

\paragraph{Bounding  $\opnorm{G(X_{\calS}, \theta) - G(X_{\calS}, \theta^*)}$ for a fixed  $\theta \in \Theta$:}
For notational convenience, define 
\begin{align}
H(X_{\calS},\theta) ~ &\triangleq 
~ G( X_{\calS}, \theta) - G( X_{\calS}, \theta^*)  \nonumber \\
& ~ = \frac{1}{\sqrt{m}}  \qth{\frac{1}{n} \sum_{i\in \calS_1}  h(X_i, \theta), ~~ \cdots, ~ \frac{1}{n} \sum_{i\in \calS_m}  h(X_i, \theta)}, \label{eq: t3}
\end{align}
where recall from \eqref{eq: t1} that the gradient difference function $h(X, \cdot)$ is defined as 
\begin{align*}
h( X, \theta) ~ = ~ \nabla f(X, \theta) - \nabla f(X, \theta^*) -   \left( \nabla F(\theta) - \nabla F(\theta^*) \right).
\end{align*}
It follows from \prettyref{ass:gradient_sub_exp} that the columns 
of $H(X_{\calS}, \theta)/\norm{\theta-\theta^*}$ 
are i.i.d.\ sub-exponential random vectors in $\reals^d$ with 
mean $0$ and scaling parameters $\sigma_2/\sqrt{nm}$
and $\alpha_2/(n\sqrt{m})$, where $\sigma_2$ and $\alpha_2$ are two absolute constants. 
Recall that $N=nm$. Applying \prettyref{thm:sub_exp_matrix} to $H(X_{\calS}, \theta)/\norm{\theta-\theta^*}$, we know that 
for any fixed $\theta$, with probability at least $1-\delta $, 
\begin{align}
\label{eq: second moment part 2}
\nonumber 
& \opnorm{H(X_{\calS}, \theta)} \\
\nonumber
&\lesssim \left(  \frac{\sigma_2}{\sqrt{n} } 
+ \frac{\sigma_2}{\sqrt{N} }   \phi \left(  d + \log \frac{1}{\delta} \right) + \frac{\alpha_2}{\sqrt{Nn} }  \phi^2 \left( d + \log \frac{1}{\delta}  \right)  
\right)  \norm{\theta-\theta^*} \\
& \lesssim \left(  \frac{1}{\sqrt{n} } 
+ \frac{1}{\sqrt{N} }   \phi \left(  d + \log \frac{1}{\delta} \right) + \frac{1}{\sqrt{Nn} }  \phi^2 \left( d + \log \frac{1}{\delta}  \right)  
\right)  \norm{\theta-\theta^*},
\end{align}
where we used $\phi(x)=\sqrt{x} \log^{3/2}(x)$,  $\sigma_2 = O(1),$ and  $\alpha_2 = O(1)$.

\paragraph{$\epsilon$-net argument:}
We apply $\epsilon$-net argument to extend the point convergence in \eqref{eq: second moment part 2} to the uniform convergence over $\Theta$. 
In particular, let 
$\calN_{\epsilon_0}$ be an $\epsilon_0$-cover
of $\Theta = \{ \theta: \norm{\theta-\theta^*} \le r \}$ with 
$$
\epsilon_0= \frac{1}{\sqrt{n}(L+L^{\prime}) }.
$$
By \cite[Lemma 5.2]{vershynin2010introduction}, we have
$$
\log \abth{\calN_{\epsilon_0}} \le d \log \left( 1 +  2r/\epsilon_0 \right)  = d\log \left( 1 + 2 r\sqrt{n} (L+L^{\prime}) \right).
$$ 
By \prettyref{eq: second moment part 2} and the union bound, we get that 
with probability at least $1-\delta$, for all $\theta \in \calN_{\epsilon_0}$, 
\begin{align}
  \opnorm{H(X_{\calS}, \theta)} \lesssim  \norm{\theta-\theta^*} \bigg( & \frac{1}{\sqrt{n} } 
+ \frac{1}{\sqrt{N} }   \phi \left(  d+ \log \frac{ \abth{\calN_{\epsilon_0} } }{\delta} \right) 
\nonumber \\
& + \frac{1}{\sqrt{Nn} }  \phi^2 \left( d+  \log \frac  { \abth{\calN_{\epsilon_0} } }{\delta} \right)  
\bigg) .   \label{en: l1}
\end{align} 
So far, we have shown the uniform convergence over net $\calN_{\epsilon_0}$.
Next, we extend this uniform convergence to the entire set $\Theta$.  

For any $\theta \in \Theta$, there exists a $\theta_k \in \calN_{\epsilon_0}$
such that $\norm{\theta - \theta_{ k }} \le \epsilon_0$. By triangle inequality,
\begin{align*}
\opnorm{ H(X_{\calS},\theta)} &  \le   \opnorm{ H(X_{\calS},\theta_k)} + \opnorm{ H(X_{\calS},\theta)-H(X_{\calS},\theta_k)}. 
\end{align*}
Note that 
\begin{align}
& \opnorm{H( X_{\calS}, \theta)-H( X_{\calS}, \theta_k)} \nonumber \\
& \le \fnorm{ H( X_{\calS}, \theta)-H( X_{\calS}, \theta_k)} \nonumber \\
& \overset{(a)}{\le} \frac{1}{n}  \max_{1 \le j \le m} \norm{  \sum_{i\in \calS_j} \left( h (X_i, \theta) - h (X_i, \theta_k) \right) } \nonumber \\
& \overset{(b)}{\le} (L+L') \norm{\theta-\theta_k} \le (L+L')\epsilon_0 = \frac{1}{\sqrt{n}},  \label{en: l2}
\end{align}
where $(a)$ follows because the Frobenius norm $\fnorm{A}^2
=\sum_j \|A_j \|^2 \le m \max_{j} \|A_j\|^2$; $(b)$ holds because
\begin{align*}
& \frac{1}{n} \norm{  \sum_{i\in \calS_j} \left( h (X_i, \theta) -   h (X_i, \theta_k) \right) }\\
& \le 
\frac{1}{n} \sum_{i\in \calS_j} \norm{h (X_i, \theta) - h (X_i, \theta_k)}
\le (L+L') \norm{\theta-\theta_k},
\end{align*}
in view of \prettyref{ass:pop_risk_smooth} and \prettyref{ass:bounded_hessian}.

Combining \eqref{en: l1} and \eqref{en: l2}, we have that with probability at least $1-\delta$, for any $\theta \in \Theta$, 
\begin{align*}
\opnorm{ H(X_{\calS},\theta)}   & \le    
\opnorm{ H(X_{\calS},\theta_k)}  + \frac{1}{\sqrt{n}} \\
& \lesssim \norm{\theta-\theta^*} \bigg(  \frac{1}{\sqrt{n} } 
+ \frac{1}{\sqrt{N} }   \phi \left(  d+ \log \frac{ \abth{\calN_{\epsilon_0 } } }{\delta} \right) \\
& \qquad + \frac{1}{\sqrt{Nn} }  \phi^2 \left( d+  \log \frac  { \abth{\calN_{\epsilon_0} } }{\delta} \right)  
\bigg)  + \frac{1}{\sqrt{n} }.  
\end{align*}
Choosing $\delta = e^{-d}$, we get that with probability at least $1-e^{-d}$, for all $\theta\in \Theta$, 
\begin{align}
 & \opnorm{ H(X_{\calS},\theta)} \nonumber \\
 & \lesssim  \norm{\theta-\theta^*} \bigg(  \frac{1}{\sqrt{n} } 
+ \frac{1}{\sqrt{N} }   \phi \left( 2d +  d \log \pth{ 1+ r  \sqrt{n} (L+L')} \right) \nonumber \\
&\qquad + \frac{1}{\sqrt{Nn} }  \phi^2 \left( 2d + d\log \pth{ 1+ r  \sqrt{n} (L+L')} \right)  
\bigg)   + \frac{1}{\sqrt{n} }. \label{en: l3}
\end{align}

\paragraph{Putting all pieces together}
Combing \eqref{en: first part} and \eqref{en: l3}, we conclude Proposition \ref{prop: second mom 1}. 

\end{proof}

\subsection*{Finish the proof of Theorem \ref{thm: uniform bound}:}
Let $\calE_1$ and $\calE_2$ denote the two events on which the conclusions in Proposition \ref{prop:uniform_converg} and  \prettyref{prop: second mom 1} hold, respectively. 
It follows from Proposition \ref{prop:uniform_converg} and  \prettyref{prop: second mom 1}
that $\prob{\calE_1 \cap \calE_2} \ge 1-3 e^{-\sqrt{d}}$. 

Recall that we use Algorithm \ref{ag: robust mean} as 
our robust gradient aggregator $\calR$ with 
input $\hat{y}_j(\theta)$ given by 
the local gradient function $g_j(\theta)$ at worker 
$j$. We will apply \prettyref{lmm:algo} with 
$y_j(\theta) = \textstyle{\frac{1}{n} \sum}_{i\in \calS_j} \nabla f (X_i, \theta)$
as per \prettyref{eq:def_y}.
Then the true mean $\mu = \nabla F(\theta)$ and 
the true sample mean $\mu_{\calS}= 
\textstyle{\frac{1}{N} \sum}_{i=1}^N \nabla f (X_i, \theta)$
as per \prettyref{eq:def_sample_mean}.

On event $\calE_1\cap \calE_2$,  in view of \prettyref{rmk: spectral triangle}, for each  $\theta \in \Theta$,
condition \prettyref{eq: bounded sample spectral} in \prettyref{lmm:algo} is satisfied with
\begin{align}
\sigma = \left( \Delta_4 + \Delta_2 \right) \norm{\theta-\theta^*} + \Delta_1 + \Delta_3,
\label{eq:choice_sigma_gradient}
\end{align}
where $\Delta_i$'s are given in 
\prettyref{prop:uniform_converg} and  \prettyref{prop: second mom 1}. 

Therefore, in view of Remark \ref{remark: sample to population mean}, 
it follows from \prettyref{lmm:algo} that for each $\theta \in \Theta,$ 
the output $G(\theta)$ 
of the gradient aggregator $\calR$
satisfies 
\begin{align*}
\norm{G(\theta) - \nabla F(\theta)} 
& \lesssim  \sqrt{ \frac{q}{m}} 
\left[ 
\left( \Delta_4 + \Delta_2 \right) \norm{\theta-\theta^*} + \Delta_1 + \Delta_3
\right] \\
& \quad + \Delta_2 \norm{\theta-\theta^*} +  \Delta_1 \\
& \lesssim 
\left( \sqrt{ \frac{q}{m} }\Delta_4 +\Delta_2 \right)  \norm{\theta-\theta^*}
+ \sqrt{\frac{q}{m} } \Delta_3 + \Delta_1.
\end{align*}

Recall that $\phi(x)=\sqrt{x} \log^{3/2}(x)$, $\log (L+L') =O(\log (Nd) )$,  
$\Theta \subset \{\theta: \norm{\theta-\theta^*} \le r \}$
for some positive parameter $r$ such that $\log r = O(\log (Nd))$, and that $N =\Omega( d^2 \log^8 (Nd))$. 
Then, 
\begin{align*}
\Delta_4 ~ &= ~  \frac{1}{\sqrt{n} } 
+ \frac{1}{\sqrt{N} }   \phi \left(  2d + d \log \pth{1 +r  \sqrt{n} (L+L')} \right) \\
& \quad + \frac{1}{\sqrt{Nn} }  \phi^2 \left(2d + d\log \pth{1 +r  \sqrt{n} (L+L')} \right)\\
& \overset{(a)}{\lesssim}  \frac{1}{\sqrt{n} } +   \sqrt{ \frac{ d }{N } } \log^2 (Nd) + \frac{1}{\sqrt{n} } ~ \lesssim ~ \frac{1}{\sqrt{n} } +   \sqrt{ \frac{ d }{N } } \log^2 (Nd),
\end{align*}
where in $(a)$ we used the assumption $N =\Omega( d^2 \log^8 (Nd))$.

Combining the last two displayed equations together with the
expressions of $\Delta_1, \Delta_2, \Delta_3$ 
in Proposition \ref{prop:uniform_converg} and  \prettyref{prop: second mom 1}, 
we get that for each $\theta \in \Theta$,
\begin{align*}
& \norm{G(\theta) - \nabla F(\theta)}  \\
& \lesssim   \left( \sqrt{ \frac{q}{N} } + \sqrt{ \frac{qd}{mN}} \log^2 (Nd)  + \sqrt{\frac{d \log (Nd)}{N}}\right)
\norm{\theta-\theta^*} \\
& \quad \; + \sqrt{ \frac{q}{N} } + \sqrt{ \frac{qd}{mN}} + \sqrt{ \frac{d}{N}} \\
& \lesssim   \left( \sqrt{ \frac{q}{N} } + \sqrt{ \frac{d}{N} } \log^2 (Nd) \right)
\norm{\theta-\theta^*}  + \sqrt{ \frac{q}{N} } + \sqrt{ \frac{d}{N}},
\end{align*}
completing the proof of \prettyref{thm: uniform bound}.

\subsection{Proof of \prettyref{thm: opt}}
\begin{proof}
From Theorem \ref{thm: uniform bound}, 
we know that there exists a constant 
$c_0$ such that with probability at least 
$1- 3e^{-\sqrt{d} }$, for all $\theta \in \Theta$, 
\begin{align}
& \norm{G(\theta)-\nabla F(\theta)}  \nonumber \\
& \le c_0  \left(  \sqrt{\frac{q}{N} } + \sqrt{ \frac{d}{N} }  \log^2 (Nd)  \right) 
 \norm{\theta-\theta^*} +  c_0 \left( \sqrt{\frac{q}{N}} + \sqrt{\frac{d }{N} } \right).
 \label{eq:gradient_diff_bound_1}
\end{align}
Thus, with probability at least $1- 3e^{-\sqrt{d} }$, we have that for every $t \ge 1$, 
\begin{align}
\norm{\theta_{t} - \theta^*} & =  \norm{\theta_{t-1} - \eta G(\theta_{t-1})- \theta^*}
\nonumber \\
&= \norm{\theta_{t-1} - \eta \nabla F(\theta_{t-1}) - \theta^* + \eta \pth{\nabla F(\theta_{t-1}) - G(\theta_{t-1})}}
\nonumber \\
& \le  \norm{\theta_{t-1} - \eta \nabla F(\theta_{t-1}) - \theta^*} + \eta\norm{\pth{\nabla F(\theta_{t-1}) - G(\theta_{t-1})}} \nonumber \\
& \le \sqrt{1-\frac{M^2}{4L^2}} \norm{\theta_{t-1} - \theta^*}+ 
\eta\norm{\pth{\nabla F(\theta_{t-1}) - G(\theta_{t-1})}}  \nonumber  \\ 
& \le \rho \norm{\theta_{t-1} - \theta^*} +  
c_0 \frac{M}{2L^2} \left( \sqrt{\frac{q}{N}} + \sqrt{\frac{d }{N} } \right),   \label{eq:noisy_gd progress} 
\end{align}
where the second inequality follows from the standard convergence analysis of perfect gradient descent,
see, \eg, \cite[Lemma 3.2]{Chen:2017:DSM:3175501.3154503};  the last inequality follows from 
\prettyref{eq:gradient_diff_bound_1}, $\eta =M/(2L^2),$ and 
\begin{align*}
\rho ~ \triangleq ~ \sqrt{1- \frac{M^2}{4L^2}} + c_0 \frac{M}{2L^2} \left(  \sqrt{\frac{q}{N} } + \sqrt{ \frac{d}{N} }  \log^2 (Nd)  \right).  
\end{align*}
Then applying a standard telescoping argument to
\prettyref{eq:noisy_gd progress} yields that
\begin{align}
\norm{\theta_{t} - \theta^*}
\le \rho^t \norm{\theta_{0} - \theta^*} 
+ \frac{c_0M}{2L^2 (1-\rho) } \left(  \sqrt{\frac{q}{N} } + \sqrt{ \frac{d}{N} }  \right).  
\label{eq:noisy_gd_bound} 
\end{align}

There exists a constant $c$ such that 
if $N\ge cd^2 \log^8 (Nd)$ and $N \ge cq$, 
then 
\begin{align*}
c_0  \left(  \sqrt{\frac{q}{N} } + \sqrt{ \frac{d}{N} }  \log^2 (Nd)  \right) \le \frac{1}{8}. 
\end{align*}
Consequently, 
\begin{align*}
\rho ~ 
& \le ~ \sqrt{1- \frac{M^2}{4L^2}} + \frac{M}{16L^2} \le 1- \frac{M^2}{8L^2}+ \frac{M}{16L^2} \le 1- \frac{M^2}{16L^2},
\end{align*}
where the last inequality follows from the assumption that $M\ge 1$. 
Combining the last displayed equation with \prettyref{eq:noisy_gd_bound} yields that 
\begin{align*}
\norm{\theta_{t} - \theta^*} 
\le \pth{1-\frac{M^2}{16L^2}}^t \norm{\theta_0-\theta^*} + 8c_0 \left( \sqrt{\frac{q}{N}} + \sqrt{\frac{d }{N} } \right), 
\end{align*}
completing the proof of Theorem \ref{thm: opt}.

\end{proof}

\section{Applications to Linear Regression and Logistic Regression}
\label{sec: linear regression}
In this section, we illustrate our general results by applying them to
the classical linear regression and logistic regression problems. 

\subsection{Application to Linear Regression}
Let $ X_i = (w_i, y_i) \in \reals^{d} \times \reals $
denote
the input data and define the risk function
$
f(X_i, \theta) = \frac{1}{2} \left( \langle w_i, \theta \rangle - y_i \right)^2.
$
For simplicity, we assume that $y_i$ is indeed generated from a linear model:
$$
y_i = \langle w_i, \theta^* \rangle + \zeta_i ,
$$
where $ \theta^* $ is an unknown true model parameter,
$ w_i \sim N(0, \identity) $ is the covariate vector whose
covariance matrix is assumed to be identity,  and $ \zeta_i \sim N(0,1) $ is i.i.d.\ additive Gaussian noise
independent of $w_i$'s. Intuitively, the inner product $ \langle w_i, \theta^* \rangle $ can be viewed as  
a linear ``measurement" of $\theta^*$ -- the signal; and $\zeta_i $ is the additive noise.

The population risk function $F$ is given by
\begin{align*}
F(\theta) &\triangleq \expect{f(X,\theta)} = \expect{  \frac{1}{2} \left( \langle w, \theta \rangle - y \right)^2}\\
&=  \expect{  \frac{1}{2} \left( \langle w, \theta \rangle - \langle w, \theta^* \rangle - \zeta \right)^2}
= \frac{1}{2}\| \theta - \theta^* \|_2^2 + \frac{1}{2},
\end{align*}
for which $ \theta^* $ is indeed the unique minimum. 
The population gradient function is
$\nabla_\theta F(\theta) = \theta - \theta^*$. It is easy to see that the population risk function $F$ is $L$-Lipschitz continuous with $L=1$, and $M$-strongly convex with $M=1$.  Hence, Assumption \ref{ass:pop_risk_smooth} is satisfied with $M=L=1$; and the stepsize $\eta=M/(2L^2) = 1/2$.

For a given random sample $X=(w, y)$, the associated random gradient is given by
$$
\nabla f(X, \theta) = w \left( \langle w, \theta \rangle - y \right) = w \langle w, \theta-\theta^* \rangle - w \zeta,
$$
where $w \sim \calN(0, \identity)$ and $\zeta \sim \calN(0,1)$ that is independent of $w$. 

The following lemma verifies that Assumption \ref{ass:gradient_sub_gaussian}--Assumption \ref{ass:gradient_sub_exp}
are satisfied with appropriate parameters.
\begin{lemma}\label{lmm:least-squares}
Under the linear regression model, the sample gradient function $\nabla f(X, \cdot)$ satisfies

(1) Assumption \ref{ass:gradient_sub_gaussian} with $\sigma_1=\sqrt{2}$
and $\alpha_1=\sqrt{2}$,

(2) and Assumption \ref{ass:gradient_sub_exp} with $\sigma_2=\sqrt{8}$ and $\alpha_2=8$. 

Moreover, with probability $1-e^{-d}$, 
Assumption \ref{ass:bounded_hessian} holds with $L'=3d + 2\log N + 2\sqrt{d (d + \log N) }$ 
for all $\{ \nabla f(X_i, \cdot) \}_{ i=1}^N$.  
\end{lemma}
\begin{proof}
Claims (1) and (2) have been proved in Lemma 4.1~\cite{Chen:2017:DSM:3175501.3154503}. It remains to prove
the last claim. Under the linear regression model
$$
\norm{\nabla f(X, \theta) - \nabla f(X, \theta')} = \norm{ww^\top (\theta-\theta')} 
\le \norm{w}^2 \norm{\theta-\theta'}.
$$
Hence, it suffices to show 
\begin{align}
\prob{ \max_{i=1}^N \norm{w_i}^2 \ge 3d + 2\log N + 2\sqrt{d (d + \log N) } } \le  e^{-d}.
\label{eq:assumption3_check} 
\end{align}
Note that for each $i \in [N]$, $\norm{w_i}^2 \sim \chi^2(d)$. Using the following tail bound:
$$
\prob{ \chi^2 (d) \ge d + 2 \sqrt{ d t} + 2 t} \le e^{-t}
$$
with $t=d +\log N $, we get that
$$
\prob{ \norm{w_i}^2 \ge 3d + 2\log N + 2\sqrt{d (d + \log N) } } 
\le \frac{1}{N} e^{- d }.
$$
Then the desired \prettyref{eq:assumption3_check} follows by the union bound. 
\end{proof}

As an immediate corollary of \prettyref{thm: opt}, our 
Byzantine-resilient Gradient Descent method can 
robustly solve the linear regression problem exponentially fast 
with high probability.

\begin{corollary}[Linear regression]
\label{cor:lr}
Under the aforementioned least-squares model for linear regression, 
assume  $\Theta \subset  \{ \theta: \| \theta-\theta^*\| \le r \}$ 
for $r>0$ such that  $\log r = O(\log (Nd))$.
Suppose that $N \ge c d^2 \log^8 (Nd))$ and $N \ge c q$ 
for a sufficiently large constant $c$, and that $m \le e^{\sqrt{d}}$.
Then with probability at least $1-3e^{-\sqrt{d}}$, the iterates $\{\theta_t\}$
given by Algorithm \ref{alg:BGD} with Robust Gradient Aggregator 
Algorithm \ref{ag: robust mean}  satisify
$$
\norm{\theta_t - \theta^* } \lesssim  \pth{\frac{15}{16}}^t \norm{\theta_0-\theta^*} 
+  \sqrt{\frac{q}{N}} + \sqrt{\frac{d }{N} } , \; \forall t \ge 0. 
$$
\end{corollary}

\subsection{Application to Logistic Regression}
Here we consider the binary logistic regression problem~\cite[Section 4.4]{statistical_learning09}, 
where we assume that $y_i$ is generated as follows:
$$
y_i \iiddistr \Bern \left( \frac{1}{1 + e^{-w_i^\top \theta^*}} \right),
$$
where $\theta^* \in \reals^d$ is an unknown true model parameter, $w_i \sim \calN(0, \sigma^2 \identity_d)$
is the observed feature vector, and $y_i \in \{0,1\}$ is the observed class label. 
Intuitively, logistric regression tries to model the log likelihood ratio 
$$
\log \frac{ \prob{y_i=1|w_i} } { \prob{y_i=0| w_i} }  
$$ 
as a linear function $w_i^\top \theta$  of $w_i$. 
Let $X_i= (w_i, y_1) \in \reals^d \times \{0,1\}$ denote the input data and define the risk function 
$f(X_i, \theta)$ as the negative log likelihood function 
\begin{align*}
f(X_i, \theta)  & = -\log \prob{ y_i \mid w_i; \; \theta} \\
& = - \indc{y_i=1} \log \frac{1}{1 + e^{-w_i^\top \theta}}
 - \indc{y_i=0}\log \frac{e^{-w_i^\top \theta}}{1 + e^{-w_i^\top \theta}}
  \\
 & = (1-y_i) w_i^\top \theta  +  \log \left( 1 + e^{-w_i^\top \theta} \right).
\end{align*}
It follows that the population risk function $F$ is given by
\begin{align*}
F(\theta) & \triangleq \expect{f (X,\theta)} \\
& = \mathbb{E}_{w \sim \calN(0, \sigma^2 \identity_d)}
\left[\frac{e^{-w^\top \theta^*} }{ 1 + e^{-w^\top \theta^*}} 
w^\top \theta  +  \log \left( 1 + e^{-w^\top \theta} \right)
\right].
\end{align*}
The population gradient function is given by
$$
\nabla_\theta F(\theta) = 
\mathbb{E}_{w \sim \calN(0, \sigma^2 \identity_d)}
\left[ w \left( \frac{e^{-w^\top \theta^*} }{ 1 + e^{-w^\top \theta^*}} 
- \frac{e^{-w^\top \theta} }{ 1 + e^{-w^\top \theta}} \right)
\right].
$$
The population Hessian function is given by
$$
\nabla^2_\theta F(\theta) = 
\mathbb{E}_{w \sim \calN(0, \sigma^2 \identity_d)} \left[ \frac{ww^\top}{ 1+ e^{w^\top \theta}} \right].
$$
It can be seen that the population Hessian function is 
strictly positive definite and hence $F(\theta)$ is strictly convex
for which $\theta^*$ is indeed the unique minimum. 
Moreover, for a given random sample $X=(w, y)$, the associated random gradient is given by
$$
\nabla f(X, \theta) = w 
\left[ (1-y)  - \frac{e^{-w^\top \theta} }{ 1 + e^{-w^\top \theta}} \right].
$$
When $\sigma$ is small and $\theta$ is restricted within a ball of small radius, 
then $w^\top \theta \sim \calN(0, \sigma^2 \|\theta\|_2^2)$ is typically small. 
In this case, we can approximate 
$e^{-w^\top \theta}/ (1 + e^{-w^\top \theta})$ by 
its first-order Taylor series $1/2- w^\top \theta/4$. As a consequence,
$$
\nabla f(X, \theta) \approx w 
\left[ \frac{1}{2} -y   +  \frac{1}{4} w^\top \theta \right],
$$
which resembles the gradient vector in the simple linear regression.

\section{Summary and Future Directions}
This present paper intersects two main areas of research:
fault-tolerant distributed computing and statistical machine learning. 
In particular, we consider a machine learning scenario where a 
model is trained in a distributed but unsecured environment. 
%
Armed with a robust mean estimation primitive, we 
secure the gradient descent method against adversarial interruptions, 
even in high dimensions. Our secured gradient descent 
converges to the true model parameter 
exponentially fast up
 to an estimation error $O(\sqrt{q/N}+\sqrt{d/N})$ -- matching 
the minimax-optimal error rate in the failure-free
setting as long as the number of faulty workers $q=O(d)$. 
A key ingredient in our analysis is a uniform concentration
of the sample covariance matrix of gradient functions. 

There are many interesting future directions to explore, and we list a few as follows. 
\begin{itemize}
\item 
We have shown the optimal error rate is $O(\sqrt{d/N})$
when $q=O(d)$. However, the optimal error rate remains elusive  when $q \gg d \gg 1$.

\item The present paper assumes the population risk function
$F(\theta)$ is convex. In many contemporary machine learning applications, the population 
risk function is often non-convex. It would be interesting to extend our results to the non-convex
setting. A crucial question is how to escape saddle points with robustly aggregated gradients.
This direction has been recently pursued in~\cite{yin2018defending}.

\item
It would be interesting to see how the choice of robust mean estimation building block affects the performance of the stochastic optimization algorithm in terms of computation, estimation error, probability error, etc.

\item 
Note that in this work, we consider full gradient descent under which each worker computes the local gradient based on the {\em entire} local sample (all $n$ data points). Since 
$n$ is small,  the computational burdens of the workers are reasonable. It has been demonstrated numerically in \cite{45648} that in the adversary-free setting, there is a performance improvement when each worker performs a few epochs of SGD before the model updates are aggregated. Whether there will be similar performance improvement in our adversary-prone setting is unclear.

\item So far, we consider synchronous distributed systems, wherein the learner communicates with the workers in synchronous communication rounds. 
It would be interesting to see how asynchrony affects the learning performance. 

\item We assume each worker reports the entire gradient vector in each round. Some applications may call for even more communication-efficient algorithms. It would be interesting to see if, rather than the entire gradient vector, it suffices for each worker to report partial gradient vector in each round.

\end{itemize}



\begin{appendices}

\section{Proof of Lemma \ref{lmm:truncation_con}}
\label{app: truncation}

We first quote a classical concentration inequality for sum of independent,
bounded random variables. 
\begin{lemma}[Bennett's inequality]
\label{lmm:Bennett}
Let $Y_1, \cdots, Y_m$ be independent random variables. Assume that $\abth{Y_j -\expect{Y_j}} \le B$ almost surely for every $j$. Then for any $t >0$, we have 
\begin{align*}
\prob{\sum_{j=1}^m (Y_j - \expect{Y_j}) \ge t} \le \exp \pth{-\frac{\sigma^2}{B^2} \cdot h\pth{\frac{Bt}{\sigma^2}}},
\end{align*}
where $\sigma^2 = \sum_{j=1}^m \var(Y_j)$ is the variance of the sum, and 
$$ h(u) =(1+u)\log(1+u) -u.$$
\end{lemma}

\begin{proof}[\bf Proof of Lemma \ref{lmm:truncation_con}]
We use the idea of truncation. In this proof, we adopt the convention that $\frac{1}{0} = +\infty$. 

For each copy $j=1, \cdots, m$, we partition $Y_j$ into countably many pieces as follows: Let 
\begin{align*}
Y_{j,0} & = Y_j \indc{|Y_j| \le 1} \\
Y_{j,k} & = Y_j \indc{ e^{k-1} \le |Y_j | \le e^{k} }, \text{ for } k=1, 2, \ldots
\end{align*}
It is easy to see that 
\begin{align*}
Y_{j} = \sum_{k=0}^{\infty} Y_{j,k}, ~~~~ \text{for } j=1, \cdots, m. 
\end{align*}

Let $S=\sum_{j=1}^m Y_j$. We have 
\begin{align*}
S=\sum_{j=1}^m Y_j = \sum_{j=1}^m \pth{\sum_{k=0}^{\infty} Y_{j,k}} = \sum_{k=0}^{\infty} \sum_{j=1}^m Y_{j,k} = \sum_{k=0}^{\infty} S_k,
\end{align*} 
where $S_k \triangleq  \sum_{j=1}^m Y_{j,k}$, for $k =0, 1, \cdots$. Thus, 
\begin{align*}
 \prob{\abth{\sum_{j=1}^m  Y_j - m \expect{Y}} > m t } 
 &= \prob{\abth{S - \expect{S}} > m t } \\
 & =\prob{\abth{\sum_{k=0}^{\infty} \pth{S_k - \expect{S_k}}} > m t }.
\end{align*}
To bound $\prob{\abth{\sum_{j=1}^m  Y_j - m \expect{Y}}> m t }$ for a given $t$, our plan is to find a sequence of $t_k$ (which depends on $t$) such that 
\begin{align}
\label{eq: truncation goal}
\sth{\abth{S - \expect{S}}> m t } ~ \subseteq ~ \cup_{k=0}^{\infty} \sth{\abth{S_k - \expect{S_k}}> m t_k},
\end{align} 
and 
$$\prob{\abth{S_k - \expect{S_k}}> m t_k}$$
is small enough to apply the union bound over all $k$.

In this proof, we choose $t_k = \frac{t}{2(k+1)^2}$ for $k=0, 1, \cdots$. It is easy to see that \eqref{eq: truncation goal} holds. 

Next, we bound $\prob{\abth{S_k - \expect{S_k}}> m t_k}$ for each $k$. 
For given $t\ge t_0$, define 
\begin{align}
\label{eq: k0}
k_0 \triangleq \inf \sth{k\in \integers: 4e^{k}(k+1)^2 \ge t}. 
\end{align}
We are particularly interested in the setting when $t\ge t_0 \ge e^2$, which implies that 
\begin{align}
\label{eq: k0 ub}
1\le k_0 \le \log t -1,
\end{align}
noting that $4e^{\log t -1}(\log t -1+1)^2 \ge t$.

\noindent {\bf Case 1}: $0\le k \le k_0-1$. It is easy to see that when $t\ge t_0 \ge e^2$, $k_0 \ge 1$. Thus, case 1 is well posed. 
As per the definition of \eqref{eq: k0}, for all $0\le k \le k_0-1$, it holds that $4e^{k}(k+1)^2 < t.$ That is, 
\begin{align}
\label{eq: lm case 1}
2e^{k} < \frac{t}{2(k+1)^2} = t_k.
\end{align}

On the other hand, by construction of $Y_{j,k}$ we have deterministically 
\begin{align}
\label{eq: truc bounded}
\abth{Y_{j,k} - \expect{Y_{j,k} }} \le 2e^{k}, ~~~ \text{for all }k. 
\end{align}
Thus
\begin{align*}
\abth{S_k - \expect{S_k}} & = \abth{\sum_{j=1}^m Y_{j,k} - \expect{\sum_{j=1}^m Y_{j,k}}}\\ 
& \le \sum_{j=1}^m \abth{Y_{j,k} - \expect{Y_{j,k} }} \le 2m e^{k} ~~~ \text{for all }k,
\end{align*}
i.e., 
$$\prob{\abth{S_k - \expect{S_k}}> 2 m e^{k}} =0  ~~~ \text{for all }k.$$ 
By \eqref{eq: lm case 1}, we have that when $0\le k \le k_0 -1$,
\begin{align}
\label{eq: small k}
\prob{\abth{S_k - \expect{S_k}}> m t_k}  \le  \prob{\abth{S_k - \expect{S_k}}> 2m e^{k}} =0.
\end{align}

\noindent {\bf Case 2:} $k_0 \le k \le  \log (mt)$. For each $k$ in this range, we will apply 
Bennett's inequality given in~\prettyref{lmm:Bennett}.

From \eqref{eq: truc bounded}, we know that for any fixed $k$, the random variable $\abth{Y_{j,k} -\expect{Y_{j,k}}}\le 2e^k$ . The variance of $Y_{j,k}$ can be bounded as follows: for $k \ge 1$
\begin{align}
\var(Y_{j,k})& \le  \expect{Y_{j,k}^2} \nonumber \\
& \le e^{2k} \prob{ \abth{Y_{j} } \ge e^{k-1}} \nonumber \\
& \le e^{2k} \exp\pth{-E\pth{e^{k-1}}}. 
 \label{eq: truc second moment}
\end{align}
For notational convenience, define 
\begin{align}
\label{def: sigma2}
\sigma_k^2 ~ \triangleq e^{2k} \exp \pth{ - E(e^{k-1})}. 
\end{align}
To see that $\sigma_k^2$ is well-defined, recall that we adopt the convention that $\frac{1}{0} =\infty$ and $\exp \pth{-\infty} =0$. 

For each $k$ in this case, i.e., $k_0\le k \le \log (mt)$, by \prettyref{lmm:Bennett}, we get 
that 
\begin{align*}
& \prob{\abth{S_k - \expect{S_k}} \ge mt_k}  \\
&= \prob{\abth{\sum_{j=1}^m (Y_{j,k} - \expect{Y_{j,k}})} \ge mt_k} \\
&\le 2\exp \pth{-\frac{\sum_{j=1}^m \var(Y_{j,k})}{e^{2(k+1)}} \cdot h\pth{\frac{e^{(k+1)} m t_k}{\sum_{j=1}^m \var(Y_{j,k})}}},
\end{align*}
Note that when $u>0$, it holds that  
$h(u) \ge u\log (u/e)$, so we have that 
\begin{align}
\label{eq: case 2 prob}
\nonumber 
& \prob{\abth{S_k - \expect{S_k}} \ge mt_k} \\
\nonumber
&\le 2\exp \pth{-\frac{\sum_{j=1}^m \var(Y_{j,k})}{e^{2(k+1)}} \cdot \frac{e^{(k+1)} m t_k}{\sum_{j=1}^m \var(Y_{j,k})}\log \pth{\frac{e^{(k+1)} m t_k}{e\sum_{j=1}^m \var(Y_{j,k})}} }\\
&=2\exp \pth{- \frac{m t_k}{e^{(k+1)}}\log \pth{\frac{e^{k} m t_k}{\sum_{j=1}^m \var(Y_{j,k})}} } \nonumber \\
& \le   2\exp \pth{- \frac{m t_k}{e^{(k+1)}}\log \pth{\frac{e^{k}t_k}{\sigma_k^2} }}, 
\end{align}
where the last inequality follows from the fact that $\sum_{j=1}^m \var(Y_{j,k}) \le m \sigma_k^2$. 
We proceed to bound $\log \pth{\frac{e^{k}t_k}{\sigma_k^2}}$ using the assumption \prettyref{eq: tail condition}:
\begin{align}
 \label{eq: log}
\nonumber 
 \log \pth{\frac{e^{k}t_k}{\sigma_k^2}}
& = \log \pth{\frac{e^{k} t}{2(k+1)^2 e^{2k} \exp \pth{-E\pth{e^{k-1}}}}}\\
\nonumber 
& = \log \pth{\frac{t}{2(k+1)^2 e^{k} \exp \pth{-E\pth{e^{k-1}}}}}\\
\nonumber 
& = \log t - \pth{\log 2 + 2\log (k+1) + k - E\pth{e^{k-1}} }\\
\nonumber 
& =  E(e^{k-1}) - \pth{\log 2 + 2\log (k+1) + k - \log t} \\
\nonumber 
& \overset{(a)}{\ge}  \frac{1}{2} E(e^{k-1})  + \pth{2k+ 4\log (k+1) + \log 2 -\log t} \\
& \quad \; \; - \pth{\log 2 + 2\log (k+1) + k - \log t} 
\nonumber \\
& \ge \frac{1}{2} E(e^{k-1}) + 2 \log (k+1) +k,
\end{align}
where inequality $(a)$ holds due to the assumption \prettyref{eq: tail condition}.
Combining the last displayed equation with \eqref{eq: case 2 prob} yields 
\begin{align}
\label{eq: case 2 prob conti}
\nonumber 
& \prob{ \abth{ S_k - \expect{S_k} } \ge mt_k} \\
\nonumber 
& \le 2 \exp \pth{- \frac{m t_k}{2 e^{(k+1)}}  E(e^{k-1}) }\\
\nonumber 
& = 2 \exp \pth{- \frac{m t}{4(k+1)^2e^{(k+1)}} E(e^{k-1})}\\
& \le 2 \exp \pth{- \frac{m t}{4(\log (mt)+1)^2e^{(k+1)}} E(e^{k-1})},
\end{align}
where the last inequality holds because in the case under consideration, $k_0\le k \le \log (mt)$. 
%
%
To proceed, we use the monotonicity assumption of $E(t)/t.$ 
If $E(t)/t $ is non-decreasing (increasing), we can bound \eqref{eq: case 2 prob conti} as 
\begin{align}
\label{eq: case 2 prob conti 1}
\nonumber 
& \prob{\abth{ S_k - \expect{S_k} } \ge mt_k} \\
\nonumber 
& \overset{(a)}{\le} 2 \exp \pth{- \frac{m t}{4( \log (mt) +1)^2e^{(k_0+1)}} E(e^{k_0-1})} \nonumber \\
& \overset{(b)}{\le} 2 \exp \pth{- \frac{m t}{4( \log (mt)+1)^2 t} E\pth{ \frac{t}{4e(k_0+1)^2}}}  
\nonumber \\
&  \overset{(c)}{\le} 2 \exp \pth{- \frac{ m }{4( \log (mt)+1)^2 } E\pth{ \frac{t}{4e \log^2 t } } },
\end{align}
where $(a)$ holds because $k_0 \le k \le \log (mt)$; 
$(b)$ holds because $k_0\le \log t-1$, $4e^{k_0} (k_0+1)^2 \ge t$, and that $E(\cdot)$ is non-decreasing; 
$(c)$ follows from $k_0 \le \log t - 1$, and that $E(\cdot)$ is non-decreasing.

If $E(t)/t$ is non-increasing, we can bound \eqref{eq: case 2 prob conti} as 
\begin{align}
\label{eq: case 2 prob conti 2}
\nonumber 
& \prob{ \abth{ S_k - \expect{S_k} } \ge mt_k} \\
\nonumber
& \le 2 \exp \pth{- \frac{m t}{4( \log (mt) +1)^2 e^{\log (mt)+1 }} 
E(e^{\log (mt)-1})} \\  
& = 2 \exp \pth{- \frac{1}{4 e(\log (mt)+1)^2} E\pth{\frac{mt}{e}}}. 
\end{align}

\noindent {\bf Case 3}: $ k \ge \log (mt).$ In this case, we use the Chebyshev's inequality:
\begin{align}
\nonumber
\prob{ \abth{  S_k - \expect{S_k} }  \ge m t_k }  & \le  \frac{ \sigma_k^2 }{t_k} 
 = \exp \pth{ - \log \frac{t_k }{\sigma_k^2} } \nonumber \\
& \overset{(a)}{\le} \frac{1}{(k+1)^2 }\exp \pth{-\frac{1}{2}E(e^{k-1})} \nonumber \\
& \le  \frac{1}{(k+1)^2} \exp \pth{-\frac{1}{2}E\pth{\frac{mt}{e}}},  
\label{eq:tail_bound_large_k}
\end{align}
where $(a)$ follows from \eqref{eq: log};
the last inequality follows from the fact that $E(u)$ is increasing (non-decreasing) in $u$.

For a fix $t$, summing over all $k \in \naturals$, we have 
\begin{align*}
& \prob{\abth{\sum_{j=1}^m Y_j - m\expect{Y}} \ge mt} \\
&\le \sum_{k=0}^{\infty} \prob{\abth{S_k- \expect{S_k}} \ge mt_k}\\
& = \sum_{k=0}^{k_0-1} \prob{\abth{S_k- \expect{S_k}} \ge mt_k} + \sum_{k = k_0 }^{
\lfloor \log (mt) \rfloor} \prob{\abth{S_k- \expect{S_k}} \ge mt_k} \\
&\quad + \sum_{\lceil \log (mt) \rceil }^{\infty}\prob{\abth{S_k- \expect{S_k}} \ge mt_k}\\
& \le 0+ \exp \pth{-\frac{1}{2}E\pth{\frac{mt}{e}}} +\sum_{k=k_0}^{\lfloor \log (mt) \rfloor} \prob{\abth{S_k- \expect{S_k}} \ge mt_k}. 
 \end{align*} 
Therefore, we have that if $E(t)/t$ is non-decreasing,
\begin{align*}
& \prob{\abth{\sum_{j=1}^m Y_j - m\expect{Y}} \ge mt}  \\
&\le 2 \log (mt)\exp \pth{- \frac{ m }{4( \log (mt)+1)^2 } E\pth{ \frac{t}{4e \log^2 t } } } \\
& \quad  + \exp \pth{-\frac{1}{2}E\pth{\frac{mt}{e}}}; 
\end{align*}
if $E(t)/t $ is non-increasing, 
\begin{align*}
& \prob{\abth{\sum_{j=1}^m Y_j - m\expect{Y}} \ge mt}  \\
&\le 2 \log (mt)\exp \pth{- \frac{1}{4 e(\log (mt)+1)^2} E\pth{\frac{mt}{e}}}\\
& \quad +  \exp \pth{-\frac{1}{2}E\pth{\frac{mt}{e}}}. 
\end{align*}

\end{proof}

\section{Tightness of Upper Bound \prettyref{eq:norm_bound_exp}}\label{app:tightness_upper_bound}

To see the upper bound \prettyref{eq:norm_bound_exp} is tight up to logarithmic factors, 
consider an example, where $A_j$'s are 
i.i.d.\  isotropic Laplace distribution with the density function given by 
$f(x)=\prod_{i=1}^d \left[ (1/\sqrt{2}) \exp \left( -\sqrt{2} x_i \right) \right]$ for $x \in \reals^d.$ In this case, note that 
$$
\left\{ \norm{A} \ge \max \{ \sqrt{m/2}, d \} \right\} \supseteq \left\{ |A_{11}| \ge d \text{ and } \sum_{j=1}^m A^2_{2j} \ge m/2\right\}.
$$
Since
\begin{align*}
\prob{ |A_{11}| \ge d } 
=
\int_{|t| \ge d}  \frac{1}{\sqrt{2} } \exp \left( -\sqrt{2} t \right) \diff t 
= \exp \left( - \sqrt{2} d \right),
\end{align*}
and by Chebyshev's inequality,
$$
\prob{ \sum_{j=1}^m A^2_{2j} \ge  m/2 } \ge  1- O(1/m) \ge \frac{1}{2}
$$
for $m$ suffciently large, and $A_{11}$ is independent of $\sum_{j=1}^m A^2_{2j}$, it follows that
\begin{align*}
\prob{  \norm{A} \ge \max \{ \sqrt{m/2}, d \} } & \ge \prob{ |A_{11} | \ge d \text{ and } \sum_{j=1}^m A^2_{2j} \ge m/2 }  \\
& \ge \frac{1}{2} \exp \left( - \sqrt{2} d \right).
\end{align*}

\section{Robust Mean Estimation}
\label{app: robust mean}

Robust gradient aggregation is closely related to robust mean estimation, formally stated next. 
 
\begin{definition}[Robust mean estimation]
Let $\calS = \sth{y_1, \cdots, y_m}$ be a sample of size $m$, wherein each of the data point $y_i$ is generated independently from an unknown distribution. 
Among those $m$ data points, up to $q = \epsilon m$ of them may be adversarially corrupted. Let $\{\hat{y}_1, \cdots, \hat{y}_m\}$ be the observed sample. The goal is to estimate the true mean of the unknown distribution when only corrupted sample $\{\hat{y}_1, \cdots, \hat{y}_m\}$ is accessible. 
\end{definition}

The adversarially corrupted data may affect the mean estimation in the following two ways:  
(i) extreme magnitudes and/or (ii) extreme directions.
The adversarial magnitudes are relatively easy to ``detected'' and removed. 
For instance, a simple trimming/pruning procedure may suffice \cite[Section 4.3.1]{7782980}. 
Dealing with adversarially extreme directions is more challenging.

If the true sample mean/center were known, then those adversarially extreme  directions would
 be ``identified'' 
by finding the eigenvectors corresponding to large eigenvalues of the sample covariance matrix;
hence the corrupted data points would be filtered away by projecting along these extremre directions. 
However, the true sample mean/center is unknown in reality. It turns out that 
we can approximate the center by representing 
each data point by sufficiently many other data points evenly, as per
Step 2 of Algorithm \ref{ag: robust mean}. 

To gain some intuition on how it works, 
let us first consider the ideal setting where the data sample is corruption-free, 
i.e., all the data points are generated from the same underlying distribution. If
the spectral norm of the sample covariance matrix is bounded, then we expect these data points
are well concentrated around the sample mean and hence "similar" to each other.

When up to an $\epsilon$ fraction of the data sample is adversarially corrupted, as long as $\epsilon$ is small enough, 
there still exists a large collection of uncorrupted data points that are close to the true sample mean and 
similar to each other. Thus, each of them can be approximately represented as a convex combination of sufficiently many other data points so that the convex coefficients are approximately uniform over this collection of data. 
Hence, the corrupted data is more responsible for the approximation error of representation. 
Thus, the direction which maximizes the approximation error of representation
is likely to the adversarially extreme  direction.

The Step 2 of Algorithm \ref{ag: robust mean} works precisely along this idea. 
In particular, \eqref{eqn: opt} aims to find good center approximation through representation,
while \eqref{eq:opt_max_min} aims to find the extreme direction.
By solving \eqref{eqn: opt} and \eqref{eq:opt_max_min} for a saddle point $(W,U)$, 
Algorithm \ref{ag: robust mean} iteratively finds a 
direction (given by $U^*$) along which all data points are spread out the most, 
and filters away data points which have large residual errors
projected along this direction (given by \prettyref{eq:residual_error}).

For completeness, next we 
present the proof of \prettyref{lmm:algo}-- the robustness guarantee of Algorithm \ref{ag: robust mean}. 

For ease of exposition, in the sequel, we let 
$$
\alpha ~ \triangleq ~  1-\epsilon \quad \text{ and } \quad  \tilde{\sigma}^2 = 2\sigma^2.
$$

We first need a minimax identity between the min-max problem \eqref{eqn: opt} and max-min problem \prettyref{eq:opt_max_min}. 
For $W \in \reals^{|\calA| \times |\calA|}$ and $U \in \reals^{d \times d}$, 
recall the function 
$\psi: (W,U) \to \reals$ defined as:
$$
\psi(W, U) = \sum_{i \in \calA}  c_i \left( \hat{y}_i - \sum_{j \in \calA} \hat{y}_j W_{ji} \right)^\top U  \left( \hat{y}_i - \sum_{j \in \calA} \hat{y}_j W_{ji} \right).
$$
Also, let $\calW$ denote the set of all column stochastic matrices $W \in \reals^{|\calA| \times |\calA|}$ such that 
$0\le W_{ji}\le \frac{4-\alpha}{\alpha (2+\alpha)m}$, and $\calU$ denote the set of all positive semi-definite 
matrices $U \in \reals^{d \times d}$ such that $\Tr(U) \le 1$. 
Then the min-max program \eqref{eqn: opt}  can be
rewritten as 
\begin{align}
W^\ast \in & \arg \min_{ W \in \calW } 
~ \max_{ U \in \calU } ~ \psi(W, U) \nonumber \\
 &=\arg \min_{ W \in \calW } ~ 
\big\| \sum_{i \in \calA}  c_i \big( \hat{y}_i - \sum_{j \in \calA} \hat{y}_j W_{ji} \big)  \big( \hat{y}_i - \sum_{j \in \calA} \hat{y}_j W_{ji} \big)^\top \big\| \label{eq:opt_W}
\end{align}
and the max-min program  \prettyref{eq:opt_max_min} can be rewritten as
$$
U^\ast \in \arg \max_{ U \in \calU } ~  \min_{ W \in \calW } ~ \psi(W, U). 
$$
Note that $\psi(W,U)$ is convex in $W$ for a fixed $U$ and concave (in fact linear) in $U$ for a fixed $W$. 
By von Neumann's minimax theorem, we have
$$
\min_{ W \in \calW } ~ \max_{ U \in \calU } ~ \psi(W, U)
= \max_{ U \in \calU } ~ \min_{ W \in \calW }  ~ \psi(W, U) = \psi(W^*, U^*).
$$
Moreover, $(W^*, U^*)$ is a saddle point, \ie, 
\begin{align}
W^*  & \in \arg \min_{W \in \calW} ~ \psi(W, U^*),   \label{eq:saddle_W} \\
 U^*  & \in \arg \max_{U \in \calU} ~ \psi(W^*, U). \label{eq:saddle_U}
\end{align}
The saddle point properties \prettyref{eq:saddle_W} and \prettyref{eq:saddle_U} are crucial to prove
\prettyref{lmm:algo}.

Moreover, by condition \prettyref{eq: bounded sample spectral}, the underlying true sample $\calS$ (of size $m$) satisfies the following condition: 
\begin{align*}
\opnorm{ \frac{1}{m} \sum_{i=1}^m (y_i - \mu_{\calS}) (y_i - \mu_{\calS})^{\top} } \le \sigma^2,
\end{align*}
where $\mu_{\calS} = \frac{1}{m} \sum_{i=1}^m y_i$.  
Recall that up to $q$ points in $\calS$ are corrupted. Let $\calS_0\subseteq \calS$ be a subset of uncorrupted subset of $\calS$ of size $m-q=(1-\epsilon) m = \alpha m $. 
Notably, since $q$ is only an upper bound on the number of corrupted data points, the choice of subset $\calS_0$ may not be unique. Nevertheless, for any choice of subset
$\calS_0$, the following holds:
\begin{align}
\label{eq: subset spectral}
\nonumber 
& \opnorm{\frac{1}{|\calS_0|} \sum_{i\in \calS_0} (y_i - \mu_{\calS})(y_i - \mu_{\calS})^{\top}} \\
\nonumber  
&= \frac{1}{|\calS_0|} \opnorm{\sum_{i\in \calS_0} (y_i - \mu_{\calS})(y_i - \mu_{\calS})^{\top}}\\
\nonumber 
& \le  \frac{1}{|\calS_0|} \opnorm{\sum_{i=1}^m (y_i - \mu_{\calS})(y_i - \mu_{\calS})^{\top}}\\ 
& \le \frac{1}{\alpha} \sigma^2  \le 2\sigma^2,
\end{align}
where the last inequality follows because by assumption, $\alpha = 1- \epsilon \ge \frac{3}{4} \ge \frac{1}{2}$.

As commented in Subsection \ref{subsec: robust agg}, Algorithm \ref{ag: robust mean} terminates within $m$ iterations. 
For ease of exposition, we use $t=1, 2, \cdots $ to denote the iteration number. 
We use $c_i(t)$, $\tau_i(t)$, and $\calA(t)$ to denote the quantities of interest at iteration $t$. 
Note that weights $c_i$ and set $\calA$ may be updated throughout an iteration. Therefore, 
we use $\calA^{\prime}(t)$ and $c_i^{\prime}(t)$ to denote the updated quantities at the end of iteration $t$. 
Note that $c_i^{\prime}(t-1) = c_i(t)$ and $\calA^{\prime}(t-1) = \calA(t)$. 

\subsection{Two auxiliary lemmas}
We first show that when Algorithm \ref{ag: robust mean} terminates, most of data points in $\calS_0$ are remained in $\calA$. 
\begin{lemma}
\label{lm: downweight}
For every iteration $t\ge 1$ in the {\bf while}--loop of Algorithm \ref{ag: robust mean}, 
\begin{align}
\label{eq: 1}
\sum_{i\in \calS_0 \cap \calA(t) } c_i(t) \tau_i (t) ~ &\le ~ \alpha m \tilde{\sigma}^2  \\
\label{eq: 2}
\sum_{i  \in \calS_0  }  
\left(        1- c_i (t)       \right) 
&  \le \frac{\alpha}{4}      \sum_{i=1}^m    \left(    1-    c_i (t)         \right)    \\
\label{eq: 3}
\abth{\calS_0\cap \calA (t) } & \ge  \frac{\alpha (2+\alpha)m}{4-\alpha}. 
\end{align}
\end{lemma}
Intuitively, Lemma \ref{lm: downweight} says that in every iteration:
(1) 
the summation of the projected  residual error over the non-corrupted data is small; (2) the weights of non-corrupted data points are reduced by a relatively small amount; (3) and more importantly, most non-corrupted data points are not removed. 
%
%
%
\begin{proof}[Proof of Lemma \ref{lm: downweight}]

The proof is by induction on \prettyref{eq: 2} and \prettyref{eq: 3}. 
Note that the induction hypotheses do not include (56).
Recall that we use $t=1, \cdots $ to denote the iteration number in the {\bf while}--loop. \\

\noindent{\bf Base case: $t=1$.} Note that 
$\calA(1)=[m]$, and $c_i(1)=1$ for all $i \in \calA(1)$.
Therefore, \prettyref{eq: 2} and \prettyref{eq: 3} hold for $t=1$ trivially. \\

\noindent{\bf Induction Step:} Suppose \prettyref{eq: 2} and \prettyref{eq: 3} hold 
for $t$, and the {\bf while}-- has not terminate at iteration $t$. We aim to show
\prettyref{eq: 2} and \prettyref{eq: 3} hold for $t+1.$

We first prove \prettyref{eq: 1} holds for $t$. 
Recall that 
$$
\tau_i(t) = \pth{  y_i-\sum_{j\in \calA(t)} \hat{y}_j W_{ji}(t) }^\top U(t) \pth{  y_i-\sum_{j\in \calA(t)} \hat{y}_j W_{ji}(t) } ,
$$
where  $W(t)$ is a minimizer to \eqref{eqn: opt} and $U(t)$ is a maximizer to \prettyref{eq:opt_max_min} at iteration $t$, respectively. 
Since $(W(t),U(t))$ is a saddle point, it follows from \prettyref{eq:saddle_W} that $W(t) \in \arg \min_{W \in \calW} \psi( W, U(t) ).$ 
Moreover, this minimization  is decoupled over all data points in $\calA(t)$ and hence each column of $W(t)$ is optimized independently. 
Therefore, by letting $W_{\ast i}(t) $ denote the column of $W(t)$ corresponding to $i \in \calA(t)$, we have
\begin{align}
W_{\ast i} (t) \in  \arg \min_{w}  &  \pth{  y_i-\sum_{j\in \calA(t) } \hat{y}_j w_j }^\top U(t) \pth{  y_i-\sum_{j\in \calA(t) } \hat{y}_j w_j } \nonumber \\
 \text{ s. t. } & \sum_{j\in \calA(t)} w_j = 1 \nonumber \\
  & 0 \le w_j \le \frac{4-\alpha}{\alpha(2+\alpha) m}. \label{eq:min_ind}
\end{align}

Let $\tilde{w} \in \reals^{\abth{\calA(t)}}$ be the column stochastic vector such that 
$$\tilde{w}_{j} \triangleq \frac{ \indc{j\in \calS_0\cap \calA(t)}}{\abth{\calS_0\cap \calA(t)}}, ~~~ \forall ~ j\in \calA(t).$$
By the induction hypothesis, $\tilde{w}$ is feasible to \prettyref{eq:min_ind}. 
Let $Y_{\calA(t)} \in \reals^{d\times n}$ be the matrix with $\hat{y}_i$ with $i\in \calA(t)$ as columns.  
Moreover, 
$$Y_{\calA(t)} \tilde{w} = \sum_{j\in \calA(t)} \hat{y}_j \tilde{w}_{j} =  
\frac{1}{\abth{\calS_0\cap \calA(t)}} \sum_{j\in \calS_0\cap \calA(t)} y_j ~ \triangleq ~ \mu_{\calS_0\cap \calA(t)}.$$ 
Thus, we have 
\begin{align*}
& \sum_{i\in \calS_0\cap \calA(t)} c_i(t) \tau_i (t) \\
& \overset{(a)}{\le} \sum_{i\in \calS_0\cap \calA(t)} c_i(t) (y_i-\mu_{\calS_0\cap \calA(t)})^{\top} U(t) \pth{y_i-\mu_{\calS_0\cap \calA(t)}}\\
& \overset{(b)}{\le} \sum_{i\in \calS_0\cap \calA(t)} (y_i-\mu_{\calS_0\cap \calA(t)})^{\top} U(t) \pth{y_i-\mu_{\calS_0\cap \calA(t)}}\\
& \overset{(c)}{\le} \sum_{i\in \calS_0\cap \calA(t)} (y_i-\mu_{\calS})^{\top} U(t) \pth{y_i-\mu_{\calS}}\\
&\le \sum_{i\in \calS_0} (y_i-\mu_{\calS})^{\top} U (t) \pth{y_i-\mu_{\calS}} \\
& \overset{(d)}{\le} \Tr\left(U(t)\right) \opnorm{\sum_{i\in \calS_0} (y_i-\mu_{\calS})(y_i-\mu_{\calS})^{\top}}\\
& \overset{(e)}{\le} ~  \alpha m \tilde{\sigma}^2,
\end{align*}
where $(a)$ holds by the optimality of $W_{\ast i}(t)$ to 
\prettyref{eq:min_ind}; $(b)$ holds because $c_i(t) \le 1$ and $U(t) \succeq 0$;
$(c)$ holds because 
$$\mu_{\calS_0\cap \calA(t)} = \frac{1}{|\calS_0 \cap \calA(t)|} \sum_{i\in \calS_0 \cap \calA(t)} y_i$$ 
is a minimizer of the quadratic form 
$$
\sum_{i\in \calS_0\cap \calA(t) } (y_i-u)^{\top} U(t) \pth{y_i-u},
$$ 
as a function of $u$; $(d)$ holds because $\abth{\iprod{A}{B}} \le \opnorm{A} \| B \|_*$, where 
$\|B\|_*$ is the sum of singlular values of $B$ and $\|B\|_*=\Tr(B)$ when $B\succeq 0$; 
$(e)$ follows by \eqref{eq: subset spectral}
and the facts that $|\calS_0|\le \alpha m$ and $\Tr(U(t))=1.$

\medskip

Next we prove \eqref{eq: 2} and \eqref{eq: 3}. Since by induction hypothesis the {\bf while}--loop has not terminate at iteration $t$,
it follows that 
\begin{align}
\sum_{i\in \calA(t)} c_i(t) \tau_i  (t) > 4 m \tilde{\sigma}^2. \label{eq:termination_imply}
\end{align}
Note that the weights of the data points that do not lie in $\calA(t)$ are not updated in iteration $t$, i.e., $c_i^{\prime}(t) =c_i(t)$ for $i\notin \calA(t)$. As a consequence, we have 
\begin{align}
& \sum_{i\in \calS_0} \left( 1-c_i^{\prime}(t) \right) \nonumber \\
& = \sum_{i\in \calS_0} \left( 1-c_i(t) \right)+ \sum_{i\in \calS_0\cap \calA(t)} \left(c_i(t) - c_i^{\prime}(t) \right) \nonumber \\ 
& \le \frac{\alpha}{4} \sum_{i=1}^m \left( 1-c_i(t)  \right) + \frac{1}{\tau_{\max}(t) } \sum_{i\in \calS_0\cap \calA(t)} \tau_i (t) c_i(t),  \label{eq: down w 111}
\end{align}
where the last inequality follows from induction hypothesis.
Furthermore, we have 
\begin{align*}
\frac{1}{\tau_{\max}(t)}\sum_{i\in \calS_0\cap \calA(t)} \tau_i (t) c_i(t) 
& \overset{(a)}{\le} \frac{1}{\tau_{\max}(t)} \alpha m \tilde{\sigma}^2 \\
& \overset{(b)}{<}  \frac{\alpha}{4\tau_{\max}(t)}\sum_{i\in \calA(t)} \tau_i(t) c_i(t), 
\end{align*}
where $(a)$ holds because we have shown that \eqref{eq: 1} holds for $t$;
$(b)$ follows from \eqref{eq:termination_imply}. 
 
Thus, \eqref{eq: down w 111} can be further bounded as 
\begin{align*}
& \sum_{i\in \calS_0} \left(1-c_i^{\prime}(t) \right)  \\
& \le \frac{\alpha}{4} \sum_{i=1}^m \left( 1-c_i(t) \right) + \frac{\alpha}{4\tau_{\max}(t)}\sum_{i\in \calA(t)} \tau_i(t) c_i(t)\\
& = \frac{\alpha}{4} \pth{\sum_{i\notin \calA(t)} (1-c_i(t)) + \sum_{i\in \calA(t)} (1-c_i(t))  + \frac{1}{\tau_{\max}(t)}\sum_{i\in \calA(t)} \tau_i(t) c_i(t) }\\
&= \frac{\alpha}{4} \pth{\sum_{i\notin \calA(t)} (1-c_i^{\prime}(t)) + \sum_{i\in \calA(t)} \pth{1-    \pth{1-\frac{\tau_i(t)}{\tau_{\max}(t)}} c_i(t)}}\\
& =  \frac{\alpha}{4}\sum_{i=1}^m (1-c_i^{\prime}(t)),
\end{align*}
proving \eqref{eq: 2} for $t+1$. We rewrite \eqref{eq: 2} for $t+1$ as 
$$
\sum_{i\in \calS_0} \left(1-c_i^{\prime}(t) \right) \le \frac{\alpha}{4-\alpha}
\sum_{i \notin \calS_0} \left(1-c_i^{\prime}(t) \right).
$$
One the one hand, we have
$$
\sum_{i \notin \calS_0} \left(1-c_i^{\prime}(t) \right) \le |\calS_0^c|  \le (1-\alpha)m.
$$
On the other hand, 
$$
\sum_{i\in \calS_0} \left(1-c_i^{\prime}(t) \right) \ge 
\sum_{i\in \calS_0 \setminus \calA'(t)} \left(1-c_i^{\prime}(t) \right)
\ge \frac{1}{2} \left| \calS_0 \setminus \calA'(t) \right|,
$$
where the last inequality holds from the fact that $c'_i(t) \le 1/2$ for all 
$i \notin \calA'(t)$ -- by the data removal criterion in Algorithm \ref{ag: robust mean}. 
Combining the last three displayed equations, we get that 
$$
\left| \calS_0 \setminus \calA'(t) \right| \le \frac{2\alpha(1-\alpha)}{4-\alpha}m,
$$
proving \eqref{eq: 2} for $t+1$.  The proof of Lemma \ref{lm: downweight} is complete. 
\end{proof}

Let $W$ be the minimizer of \eqref{eqn: opt} when the {\bf while}--loop terminates. 
Let $W_1$ be the result of zeroing out all singular values of $W$ that are greater than 0.9. 

\begin{lemma}
\label{lmm: rank one}
The matrix $W_0=(W-W_1)(I-W_1)^{-1}$ is a column stochastic matrix, and the rank of the weight matrix  $W_0$ is one. 
\end{lemma}
\begin{remark}\label{rmk:Z_column}
Let $X_{\calA}\subseteq \reals^{d\times \abth{\calA}}$ be the data matrix with columns 
being the data points in $\calA$.  Let $Z=X_{\calA} W_0$. 
Since $W_0$ is rank one, all the $|\calA|$ columns in the matrix $Z$ are identical. Denote 
\begin{align}
\label{def: z identical}
Z = \qth{\tilde{\mu}, \cdots, \tilde{\mu}}. 
\end{align}
Then $\tilde{\mu}$ is a weighted average of the points in $\calA$.   
\end{remark}

\begin{proof}
We first show that $W_0$ is a column stochastic matrix:
\begin{align*}
\bm{1}^{\top} W_0 &=  \bm{1}^{\top} (W-W_1)(I-W_1)^{-1} \overset{(a)}{=} (\bm{1}^{\top}- \bm{1}^{\top}W_1)(I-W_1)^{-1} \\
& =  \bm{1}^{\top} (I-W_1)(I-W_1)^{-1} = \bm{1}^{\top}, 
\end{align*}
where $(a)$ follows because $W$ is column stochastic. 

Next we show that rank of $W_0$ is one. 
From \eqref{eqn: opt}, we know that $\Fnorm{W}^2 \le \frac{4-\alpha}{\alpha (2+\alpha)}$. To see this, 
\begin{align*}
\Fnorm{W}^2 & = \sum_{i,j\in \calA} W_{ji}^2  \\
& \le \sum_{i,j\in \calA} \pth{W_{ji} \cdot \max_{i,j\in \calA}W_{ji} } \\
& \le  \pth{\sum_{i,j\in \calA} W_{ji}  } \frac{4-\alpha}{\alpha (2+\alpha)m} \\
& \le   \frac{4-\alpha}{\alpha (2+\alpha)}. 
\end{align*}
When $\alpha \ge \frac{3}{4}$, 
$$
\frac{4-\alpha}{\alpha (2+\alpha)} \le \frac{52}{33} < 2\times 0.9^2.
$$
Hence, at most one singular value of $W$ can be greater than $0.9$. 
Moreover, since $W$ is column stochastic, its largest singular value is at least $1.$
Thus, $W-W_1$ is of rank one. As a consequence, $W_0$ is of rank one.  
\end{proof}

\subsection{Proof of \prettyref{lmm:algo}}

\begin{proof}
Recall that our goal is to show 
\begin{align*}
\norm{\mu_{\calS} -\hat{\mu}} = O(\sigma \sqrt{1-\alpha}),
\end{align*}
where $\hat{\mu} =\frac{1}{\abth{\calA}} \sum_{i \in \calA} \hat{y}_i$ is the algorithm output.
Recall $Y_{\calA}\subseteq \reals^{d\times \abth{\calA}}$ is the data matrix with columns 
being the data points in $\calA$. 
In view of \prettyref{rmk:Z_column}, columns of $Z=Y_{\calA} W_0$ are identical and denoted by
$\tilde{\mu}$.
Our proof is divided into two steps:

{\bf Step 1}: We first show that points in $\calA$ are clustered around the center $\tilde{\mu}$. 
In addition, by \eqref{eq: 3} in Lemma \ref{lm: downweight}, the set $\calA$ mainly consists of uncorrupted data. As a consequence,
we are able to show that 
\begin{align}
\label{al: aux 111}
\hat{\mu} ~= \frac{1}{|\calA|} \sum_{i\in \calA} \hat{y}_i ~ \approx  \frac{1}{|\calS_0 \cap \calA|} \sum_{i \in \calS_0 \cap \calA} \hat{y}_i = \frac{1}{|\calS_0 \cap \calA|} \sum_{i \in \calS_0 \cap \calA} y_i .  
\end{align}

{\bf Step 2}: By \prettyref{eq: subset spectral}, points in $\calS_0$ are clustered around the center $\mu_{\calS}$. 
In addition, by \eqref{eq: 3} in Lemma \ref{lm: downweight}, most of the points in $\calS_0$ have been preserved.  Thus we  
are able to show that 
\begin{align}
\label{al: aux 222}
\mu_{\calS} = \frac{1}{m} \sum_{i =1}^m y_i ~ \approx ~  \frac{1}{|\calS_0 \cap \calA|} \sum_{i\in \calS_0 \cap \calA} y_i.  
\end{align} 

Putting these two pieces together, the proof of \prettyref{lmm:algo} is complete.

\medskip 

\noindent {\bf Step 1: We show \eqref{al: aux 111}. } 

When the {\bf while}--loop terminates, in view of \prettyref{eq:opt_W}, we have 
\begin{align}
\opnorm{Y_{\calA}(I-W) \diag{(c_{\calA})^{\frac{1}{2}} }} \le 2 \sqrt{m} \tilde{\sigma}. \label{eq:obj_1},
\end{align}
where $\diag{ (c_{\calA})^{\frac{1}{2} } }$ is the diagonal matrix with diagonal entries given by $\{  c_i^{1/2} \}_{i \in \calA}.$
%
%
We will show that $\hat{y}_i\approx  \tilde{\mu}$ for all $i \in \calA$. 
For this purpose, it is enough to show $\opnorm{Y_{\calA} - Z}$ is small: 
\begin{align*}
\opnorm{Y_{\calA} - \tilde{\mu} \bm{1}^T} & = \opnorm{Y_{\calA} - Z} \\
& = \opnorm{Y_{\calA} - Y_{\calA} W_0} \\
&= \opnorm{Y_{\calA}(I-W_1) (I-W_1)^{-1}- Y_{\calA} (W-W_1)(I-W_1)^{-1} }\\
& = \opnorm{Y_{\calA}(I-W)(I-W_1)^{-1} }\\
& \le \opnorm{Y_{\calA}(I-W)}  \opnorm{(I-W_1)^{-1}}  \\
& \overset{(a)}{\le} \opnorm{Y_{\calA}(I-W)} \times 10 \\
& \overset{(b)}{\le} 10 \sqrt{2} \opnorm{Y_{\calA}(I-W) \diag{(c_{\calA})^{\frac{1}{2}} }}\\
& \overset{(c)}{\le} 20\sqrt{2m}\tilde{\sigma},
\end{align*}
where $(a)$ holds because the largest singular value of $W_1$ is at 
most $0.9$; $(b)$ holds because $c_i \ge \frac{1}{2}$ for all $i \in \calA$;
$(c)$ follows from \prettyref{eq:obj_1}.

Fix any $0 < \epsilon^{\prime} < 1/2$. Let $\calT\subseteq \calA$ such that $\abth{\calT} \ge (1-\epsilon^{\prime}) |\calA|$.  
We have 
\begin{align}
\label{eq: modified aaa}
\nonumber
& \norm{\frac{1}{|\calT|} \sum_{i\in \calT} \hat{y}_i  ~ - ~ \hat{\mu}} \\
\nonumber
& =  \norm{\frac{1}{|\calT|}\sum_{i\in \calT} \hat{y}_i -\frac{1}{|\calA|} \sum_{i\in \calA} \hat{y}_i} \\
\nonumber 
& = \norm{\frac{1}{|\calT|}\sum_{i\in \calT} (\hat{y}_i -\tilde{\mu})-\frac{1}{|\calA|} \sum_{i\in \calA} (\hat{y}_i -\tilde{\mu})}\\
\nonumber
&= \norm{\pth{\frac{1}{|\calT|} - \frac{1}{|\calA|}}\sum_{i\in \calT} (\hat{y}_i -\tilde{\mu})  -\frac{1}{|\calA|} \sum_{i\in \calA /\calT} (\hat{y}_i -\tilde{\mu}) }\\
\nonumber
& \overset{(a)}{\le} \frac{|\calA|-|\calT|}{|\calT||\calA| } \norm{\qth{Y_{\calA} - Z}_{\calT} \bm{1}} + \frac{1}{|\calA|}\norm{\qth{Y_{\calA} - Z}_{\calA /\calT} \bm{1}} \\
\nonumber
& = \pth{\frac{|\calA|-|\calT|}{\sqrt{|\calT|} \, |\calA| } +\frac{\sqrt{\abth{\calA /\calT}}}{|\calA|} } \opnorm {Y_{\calA} - Z} \\
& \le 80 \sqrt{2}  \tilde{\sigma} \sqrt{ \epsilon^{\prime} } ,
\end{align}
where $\qth{Y_{\calA} - Z}_{\calT} $ denotes the submatrix of $Y_{\calA} - Z$ -- restricting to columns in $\calT$, and $\bm{1}\in \reals^{\abth{\calT}}$; the last inequality holds because $\epsilon'<1/2$ and
$$
|\calA| \ge |\calA\cap \calS_0| \ge \frac{\alpha(2+\alpha)}{4-\alpha} m.
$$
%
%
Note that 
$$
\frac{\alpha(2+\alpha)}{4-\alpha} \ge 1-\frac{5}{3}(1-\alpha) \Leftrightarrow (\alpha-1)^2 \ge 0.
$$
Thus, $\abth{\calA - \calA\cap \calS_0} \le \frac{5}{3}(1-\alpha) m.$
Choosing $\calT = \calA\cap \calS_0$, we obtain 
\begin{align}
\label{estimator b}
\norm{\mu_{\calS_0\cap \calA} - \hat{\mu}} \le  80\sqrt{2}\tilde{\sigma} \sqrt{5(1-\alpha)/3} \le 160 \tilde{\sigma} \sqrt{1-\alpha} = O(\tilde{\sigma} \sqrt{1-\alpha}). 
\end{align}

\medskip
\noindent {\bf Step 2: We show \eqref{al: aux 222}. }   The proof of \eqref{al: aux 222} is similar to that of \eqref{al: aux 111}. 

Recall that $\mu_{\calS} =\frac{1}{m} \sum_{i=1}^m y_i$ and that 
$$
\mu_{\calS_0\cap \calA} =\frac{1}{|\calS_0\cap \calA|} \sum_{i\in \calS_0\cap \calA} y_i.
$$ 
We have 
\begin{align}
\label{eq: refine 1}
\nonumber 
\norm{\mu_{\calS} -\mu_{\calS_0\cap \calA}} & = \norm{ \mu_{\calS} - \frac{1}{\abth{\calA\cap \calS_0}}\sum_{i \in \calS_0\cap \calA } y_i} \\
\nonumber 
& = \norm{  \frac{1}{\abth{\calA\cap \calS_0}}\sum_{i \in \calA\cap \calS_0} (y_i  -\mu_{\calS})} \\
\nonumber 
& =\frac{1}{\abth{\calA\cap \calS_0}} \opnorm{  [ Y_{\calA \cap \calS_0} - \mu_S ] \mathbf{1} } \\
& \le \frac{\sqrt{ \abth{\calS_0} } } { \sqrt{ \abth{ \calA\cap \calS_0}}  } \tilde{\sigma} \nonumber \\
& \le \sqrt{\frac{4-\alpha}{\alpha(2+\alpha)} } \sqrt{1-\alpha} \tilde{\sigma} \le \sqrt{2 (1-\alpha)} \tilde{\sigma} \nonumber .
\end{align}

\end{proof}

\subsection{Alternative Termination Condition of Algorithm \ref{ag: robust mean}}
\label{app: ag modification}
Recall that the termination of Algorithm \ref{ag: robust mean} relies on the knowledge of $\sigma$. 
Since $\sigma$ depends on $\norm{\theta - \theta^*}$ according to
\prettyref{eq:choice_sigma_gradient} 
in our robust gradient aggregation setting, 
the learner needs to know a priori $\norm{\theta - \theta^*}$ for all $\theta$, which may not be possible. 
Nevertheless, it turns out that the termination condition of Algorithm \ref{ag: robust mean} can be replaced by 
checking the cardinality of set $\abth{\calA}$, formally stated as follows: 

If 
  $$
  \abth{\calA \setminus \sth{i: ~ \pth{1-\frac{\tau_i}{\tau_{\max}}} c_i \le \frac{1}{2}}} \ge \frac{\alpha(2+\alpha)m}{4-\alpha},
  $$ 
  we update $c_i \gets \pth{1-\frac{\tau_i}{\tau_{\max}}} c_i$
  and remove $\sth{i: ~c_i\le \frac{1}{2}}$ from $\calA$; otherwise, we break the {\bf while}--loop.  

Similar to the original Algorithm \ref{ag: robust mean}, in the modified Algorithm \ref{ag: robust mean}, in each iteration of the {\bf while}--loop at least one point will be removed. Thus, the modified Algorithm \ref{ag: robust mean} terminates in at most $m$ iterations. We next prove the conclusion of \prettyref{lmm:algo}
still holds after this modification.

Suppose the modified Algorithm \ref{ag: robust mean} terminates at iteration $t^*$. By the modified code we know $\abth{\calA(t^*)} \ge \frac{\alpha(2+\alpha)m}{4-\alpha}$; otherwise, the algorithm terminates earlier than $t^*$. By the termination condition, we also know that 
\begin{align}
\abth{\calA(t^*) - \sth{i: ~ \pth{1-\frac{\tau_i}{\tau_{\max}}} c_i \le \frac{1}{2}}} < \frac{\alpha(2+\alpha)m}{4-\alpha}.
\end{align}
%
\begin{claim}
There exists an iteration $t^{\prime} \le t^*$ such that 
$$
\sum_{i\in \calA(t^{\prime})} c_i(t^{\prime}) \tau_i(t^{\prime}) \le 8 m\sigma^2.
$$ 
\end{claim}
\begin{proof}
We prove by contradiction. Suppose 
\begin{align}
\sum_{i\in \calA(t) } c_i(t) \tau_i (t) > 8 m\sigma^2,  
\quad \forall t \le t^* .   
\end{align}
Note that the modified Algorithm \ref{ag: robust mean} and the original Algorithm \ref{ag: robust mean} differ only in their termination conditions. 
Recall that the original termination condition is only used in the proof of Lemma \ref{lm: downweight}  to
conclude that \prettyref{eq:termination_imply} holds when the {\bf while}-loop does not terminate. 
Thus, under the hypothesis (given in the last displayed equation),  
Lemma \ref{lm: downweight} still holds. It follows that 
\begin{align*}
& \abth{\calA(t^*) - \sth{i: ~ \pth{1-\frac{\tau_i}{\tau_{\max}}} c_i \le \frac{1}{2}}}  \\ 
& \ge \abth{\calS_0\cap\pth{\calA(t^*) - \sth{i: ~ \pth{1-\frac{\tau_i}{\tau_{\max}}} c_i \le \frac{1}{2}}}}\\
&\ge \frac{\alpha(2+\alpha)m}{4-\alpha},
\end{align*}
%
which leads to a contradiction.
\end{proof}

Since $\calA(t)$ is monotone decreasing, it follows that $\calA(t^*) \subseteq  \calA(t')$.
Moreover, 
\begin{align*}
 \abth{\calA(t^*)} \ge \frac{\alpha(2+\alpha)m}{4-\alpha}
 \ge   \frac{\alpha(2+\alpha)}{4-\alpha} \abth{\calA(t')} \ge \pth{1-\frac{5}{3} (1-\alpha)} \abth{\calA(t')}.
\end{align*}
By \eqref{eq: modified aaa}, we know 
\begin{align*}
& \norm{\frac{1}{\abth{\calA(t^*)}} \sum_{i\in \calA(t^*)}\hat{y}_i   - \frac{1}{\abth{\calA(t')}} \sum_{i\in \calA(t')}\hat{y}_i} \\
& \le 80\sqrt{2} \tilde{\sigma} \sqrt{\frac{5}{3}(1-\alpha)} = O(\sigma \sqrt{1-\alpha}). 
\end{align*}
From Lemma \ref{lmm:algo}, we know 
$$
\norm{\frac{1}{\abth{\calA(t')}} \sum_{i\in \calA(t')}\hat{y}_i - \mu_{\calS}} = O(\sigma \sqrt{1-\alpha}).$$ 
Combining the last two displayed equations, we have 
\begin{align*}
\norm{\frac{1}{\abth{\calA(t^*)}} \sum_{i\in \calA(t^*)}\hat{y}_i   - \mu_{\calS}} = O(\sigma \sqrt{1-\alpha}). 
\end{align*}

\end{appendices}

\section*{Acknowledgment}
L. Su was supported  in part by the NSF Science \& Technology Center 
for Science of Information Grant CCF-0939370.
J.~Xu was supported in part by the NSF Grant CCF-1755960.
\bibliographystyle{alpha}
\bibliography{BGDE,ARB}

\end{document}